%% file: submit.tex
\documentclass[10pt,journal]{IEEEtran}

\input{preamble.tex}

\usepackage[mathx]{mathabx}
\usepackage{dsfont}
\usepackage{breqn}

\def\bmd{{\bm d}}
\def\bmg{{\bm g}}
\def\bmn{{\bm n}}
\def\bmq{{\bm q}}
\def\bmw{{\bm w}}
\def\bmu{{\bm u}}
\def\bmv{{\bm v}}
\def\bmr{{\bm r}}
\def\bmR{{\bm R}}

\newcommand{\hbmw}{\widehat{\bmw}}

\def\msd{{\mbox{MSD}}}
\def\net{{\mbox{net}}}

\newcommand{\tbpsi}{\widetilde{\bpsi}}
\newcommand{\tbmw}{\widetilde{\bmw}}

\newcommand{\ft}[1]{\scriptstyle{\textrm{#1}}}
\newcommand{\fr}[1]{\scriptstyle{\cref{#1}}}

\begin{document}

\title{Privacy-Preserving Distributed Projection LMS for Linear Multitask Networks}
\author{Chengcheng~Wang,~\IEEEmembership{Member,~IEEE}, Wee~Peng~Tay,~\IEEEmembership{Senior~Member,~IEEE}, Ye~Wei,~\IEEEmembership{Member,~IEEE}, and Yuan~Wang%
	\thanks{This work was supported in part by the Singapore Ministry of Education Academic Research Fund Tier 2 grant MOE2018-T2-2-019 and by A*STAR under its RIE2020 Advanced Manufacturing and Engineering (AME) Industry Alignment Fund – Pre Positioning (IAF-PP) (Grant No. A19D6a0053).}%
\thanks{C. Wang, W. P. Tay and Y. Wang are with Nanyang Technological University, Singapore. Y. Wei is with Northeast Electric Power University, China. E-mails: \texttt{wangcc@ntu.edu.sg}; \texttt{wptay@ntu.edu.sg}; \texttt{weiye@hrbeu.edu.cn}; \texttt{wangyuan@ntu.edu.sg}. 
}
}                                                                                                           
\maketitle 

\begin{abstract}       
We develop a privacy-preserving distributed projection least mean squares (LMS) strategy over linear multitask networks, where agents' local parameters of interest or tasks are linearly related. Each agent is interested in not only improving its local inference performance via in-network cooperation with neighboring agents, but also protecting its own individual task against privacy leakage. In our proposed strategy, at each time instant, each agent sends a noisy estimate, which is its local intermediate estimate corrupted by a zero-mean additive noise, to its neighboring agents. We derive a sufficient condition to determine the amount of noise to add to each agent's intermediate estimate to achieve an optimal trade-off between the network mean-square-deviation and an inference privacy constraint. We propose a distributed and adaptive strategy to compute the additive noise powers, and study the mean and mean-square behaviors and privacy-preserving performance of the proposed strategy. Simulation results demonstrate that our strategy is able to balance the trade-off between estimation accuracy and privacy preservation.
\end{abstract}                                                       
\begin{IEEEkeywords}   
Distributed strategies, multitask networks, inference privacy, privacy preservation, additive noises.
\end{IEEEkeywords} 

\section{Introduction}\label{Sect: introduction}
In multitask distributed networks, a set of interconnected agents work collaboratively to estimate different but related parameters of interest \cite{CheRicSay:J14,CheRicSay:J15,LengJSP0115,PlaBogBer:J15,NasRicFer:J16,CheRicSay:J17,NasRicFer:J17,WanTayHu:J17}. In order to make use of the relationship between different tasks for better inference performance, local estimates are exchanged amongst agents within the same neighborhood. However, each agent may wish to protect its own local parameters of interest and prevent other agents in the network from accurately inferring these parameters. For example, in an Internet of Things (IoT) network, sensors are deployed in smart grids, traffic monitoring, health monitoring, home monitoring and other applications \cite{ChanCampoEsteveEtAl2009,whitmore2015internet,shah2016survey,li20185g}. Although different IoT or edge computing devices may have their local objectives, they can exchange information with each other or service providers \cite{HuTay:J15,TayTsiWin:J08b,Tay:J15,HoTayQue:J15} to improve inferences and services. However, sharing information may lead to privacy leakage.

To protect the privacy of data being exchanged between agents in a distributed network, the works \cite{wang2003using,xiong2016randomized,sarwate2014rate,liao2017hypothesis,imtiaz2018differentially,ren2018textsf} propose local differential privacy mechanisms, while \cite{razeghiprivacy,mcsherry2009privacy,dwork2006our,duchi2013local} develop privacy-preserving distributed data analytics. However, these approaches often lead to significant trade-offs in estimation accuracy as they do not specifically protect the privacy of the parameters of interest. Privacy-preserving average consensus algorithms in \cite{MoMur:J17,NozTalCor:J17,HeCaiGua:J18,HeCaiZha:J19,WanHeChe:J19} aim to guarantee the privacy of initial states while achieving asymptotic consensus on the average of the initial values. To achieve inference privacy in a decentralized IoT network, \cite{SunTayHe:J18,HeTayHua:J19,SunTay:J20a,SunTay:J20b,WanSonTay:J20} propose approaches with information privacy guarantees, while \cite{diamantaras2016data,KunSPM2017,KunJFI2017,SonWanTay:C18,WanYanTay:C18,SonWanTay:J20,LauTay:C20} map the agents' raw observations into a lower dimensional subspace. These inference privacy works assume that all agents in the network are interested in inferring the same parameters or hypothesis of interest. 

Existing works on multitask distributed networks mainly focus on developing new distributed strategies for automatic clustering \cite{zhaSay:J15,CheRicSay:J15} and handling different relationships among the tasks \cite{CheRicSay:J14,PlaBogBer:J15,NasRicFer:J16,CheRicSay:J17,NasRicFer:J17,WanTayHu:J17}. Other works like \cite{CheSay:J13,Say:J14,NasRicChe:C16,CheRicSay:J15} focus on evaluating the performance of different distributed schemes. Few works have considered protecting the privacy of each agent's local parameters. The reference \cite{HarFlaRic:C16} considers data privacy of the agents' local measurements in a single-task network. 

Our objective is to develop a privacy-preserving distributed projection least mean squares (LMS) strategy over multitask networks, which balances the trade-off between estimation accuracy and privacy preservation of agents' local parameters or tasks. Specifically, we consider multitask estimation problems where the unknown parameters of interest within each neighborhood are linearly related with each other \cite{NasRicFer:J17}. Such problems exist widely in applications such as electrical networks, telecommunication networks, and pipeline networks. Different from the strategy in \cite{NasRicFer:J17}, which does not take privacy preservation into consideration, we propose to sanitize each agent's intermediate estimate before sharing it with its neighbors by adding an appropriate zero-mean noise to the intermediate estimate. We study how to design the power of the noise to be added to optimize the trade-off between the network mean-square-deviation (MSD) and the inference privacy of each agent's local parameters, measured by its neighbors' mean-square error in estimating the agent's local parameters. In addition, different from existing works on distributed strategies in the presence of link noises \cite{Say:J14,NasRicChe:C16,KhaTinRas:J12,ZhaTuSay:J12}, which examine the effect of link noises on the performance of distributed strategies, this work focuses on developing agent-specific time-varying variances of the additive noises that enable agents to benefit from in-network cooperation as well as to protect their own individual tasks against privacy leakage.

The rest of this paper is organized as follows. In \cref{Sect: Multitask Network}, we introduce the system assumptions and the multitask estimation problem considered in this paper. In \cref{Sect: Privacy-preserving Diffusion Strategy}, we propose a privacy-preserving multitask distributed strategy by introducing privacy noises and show how to obtain the noise powers in \cref{Sect: Privacy Mechanism Noise Design and Convergence Analysis}. We examine the mean and mean-square behaviors and privacy-preserving performance of the proposed strategy in \cref{Sect: Performance Analysis}. We present simulation results in \cref{Sec:Simulation results}. \cref{Sect: Conclusion} concludes the paper.

\emph{Notations:} We use lowercase letters to denote vectors and scalars, uppercase letters for matrices, plain letters for deterministic variables, and boldface letters for random variables. We use $(\cdot)\T$, $(\cdot)^{-1}$, $\trace{\cdot}$ and $\rank\parens*{\cdot}$ to denote matrix transpose, inversion, trace and rank, respectively. We use $\diag\{\cdot\}$, $\col{\cdot}$ and $\row{\cdot}$ for a diagonal matrix, column vector and row vector, respectively. The $M\times M$ identity matrix is denoted as $I_M$ and $0_{M\times N}$ denotes an $M\times N$ zero matrix. We have $\norm{x}^2_\Sigma=x\T\Sigma x$ with $\Sigma$ being a symmetric positive semidefinite matrix, and let $\norm{x}^2=\norm{x}_{I_M}^2$ where $M$ is the length of $x$. In addition, we use $\rho\parens{A}$ to denote the spectral radius of matrix $A$, and $\lambda_{\min}(A)$ and $\lambda_{\max}(A)$ for the minimum and maximum eigenvalues of the symmetric matrix $A$, respectively. For a block matrix $A$ whose $(k,\ell)$-th block represents some interaction between agents $k$ and $\ell$, we let $[A]_k$ and $[A]_{\cdot,k}$ denote the block row and block column of $A$ corresponding to agent $k$, respectively, and $[A]_{k,\ell}$ be the $(k,\ell)$-th block.

\section{Linear Multitask Networks}\label{Sect: Multitask Network}

In this section, we present our system model, and give a brief introduction to multitask networks, where neighboring agents' tasks are linearly related. Consider a strongly-connected network of $N$ agents, where information can flow in either direction between any two connected agents \cite{Say:J14}. At each time instant $i\geq0$, each agent $k$ has access to a scalar observation $\bmd_k(i)$, and an $M_k\times 1$ regression vector $\bmu_{k}(i)$. The random data $\{\bm d_k(i),\bm u_{k}(i)\}$ are related via the linear regression model 
\begin{align}\label{model}
\bmd_k(i)=\bmu_{k}\T(i)\bmw_k^o+\bmv_k(i)
\end{align}
where the scalar $\bm v_k(i)$ is measurement noise, $\bmw_k^o$ is an $M_k\times 1$ unknown random vector, with mean $\E\bmw_k^o$ and covariance matrix 
\begin{align}\label{Wkk}
W_{kk} = \E[(\bmw^o_k-\E\bmw^o_k)(\bmw^o_k-\E\bmw^o_k)\T].
\end{align}
Although we assume that the parameter vector $\bmw_k^o$ is random instead of being a deterministic parameter vector, like most of the literature on distributed strategies \cite{CheRicSay:J15,CheRicSay:J17,CheRicSay:J14,CheSay:J13,HarFlaRic:C16,KhaTinRas:J12,NasRicChe:C16,NasRicFer:J17,NasRicFer:J16,PlaBogBer:J15,ZhaSay:C12,zhaSay:J15,ZhaTuSay:J12,Say:J14}, we assume that the parameter vector $\bmw_k^o$ is fixed at a certain realization $w_k^o$ during the distributed estimation process. Since our goal is to develop inference privacy mechanisms that lead to high estimation errors, on average, of agent $k$'s local parameters $\bmw_k^o$ by other agents $\ell\ne k$, we adopt a Bayesian framework for the privacy criterion.

We make the following assumptions regarding model (\ref{model}).

\begin{Assumption}\label{Ass:measurement_noise}(Measurement noise)
The measurement noise $\bm v_k(i)$ is white over time, with zero mean, and a variance of $\sigma_{v,k}^2$. 
\end{Assumption}
\begin{Assumption}\label{Ass:Regression_data}(Regressors)
The regressors $\{\bm u_{k}(i)\}$ are zero-mean, white over time and space with
\begin{align}\label{Ruk}
\E\bm u_{k}(i)\bm u_{\ell}\T(j)= R_{u,k}\delta_{k,\ell}\delta_{i,j},
\end{align}
where $R_{u,k}$ is symmetric positive definite, and $\delta_{k,\ell}$ is the Kronecker delta.
\end{Assumption}
\begin{Assumption}\label{Ass:Independence_assumption}
(Independence) The random data $\{\bmw_k^o,\bmu_\ell(i),\bmv_m(j)\}$ are independent of each other for any agents $k,\ell,m$ and any time instants $i,j$.
\end{Assumption}

For a realization $\{\bmw_k^o=w_k^o \mid k=1,\ldots,N\}$, the objective of each agent $k$ is to find the minimizer of the following mean-square-error cost function:
\begin{align}\label{Eq:ind_cost}
J_k(w_k) = \E[(\bmd_k(i)-\bmu_k\T(i)w_k)^2]{\bmw_k^o=w_k^o}.
\end{align}
Let $\calN_k$ be the set of all neighbors of agent $k$, including agent $k$ itself, and 
\begin{align*}
\bmw_{\calN_k}^o=\col{\bmw_\ell^o}_{\ell\in\calN_k}.
\end{align*}
Following \cite{NasRicFer:J17}, we assume that neighboring tasks $\{\bmw_\ell^o \mid \ell\in\calN_k\}$ are related via 
\begin{align}\label{Eq:local_const}
\calD_{k} \bmw^o_{\calN_k} + b_k =0,
\end{align}
where the matrix $\calD_k$ is a $j_k\times |\calN_k|$ block matrix, and the vector $b_k$ is a $j_k\times 1$ block vector, with $j_k\geq1$. Then, the objective for the entire network is to find the optimal solution to the following constrained optimization problem \cite{NasRicFer:J17}: 
\begin{subequations}\label{Eq:J_opt}
\begin{align}
\min_{w_1,\ldots,w_N}\ & \sum_{k=1}^N J_k(w_k)\\
\st\ & w_{\calN_k}=\col{w_\ell}_{\ell\in\calN_k},\\
& \calD_{k} w_{\calN_k} + b_k =0,\ \text{for $k=1,\ldots,N$,}\label{Eq:cost_constraints} 
\end{align}
\end{subequations}
where individual costs $\{J_k(w_k)\}$ are defined by \cref{Eq:ind_cost}. We assume that the optimization problem \cref{Eq:J_opt} is feasible \cite{NasRicFer:J17}. We assume that the matrix $\calD_k$ for any agent $k$ is full row-rank so that the linear equation \cref{Eq:cost_constraints} has at least one solution. To avoid having a trivial solution, we also assume that it is column-rank deficient because otherwise \cref{Eq:cost_constraints} has a unique solution.

As demonstrated in \cite{NasRicFer:J17}, each agent $k$ benefits through cooperation with neighboring agents by sharing its local parameter estimate $\bpsi_k(i)$ with its neighbors at each time instant $i$. By leveraging the linear relationships \cref{Eq:local_const} and its neighbors' parameter estimates, an agent can improve its own inference accuracy. In this paper, we consider the scenario where agent~$k$ also wants to prevent other agents from inferring its own task $\bmw_k^o$. Thus, a privacy-preserving distributed solution is required to balance the trade-off between estimation accuracy and privacy protection of the individual tasks. In order to limit privacy leakage, we add a zero-mean, \gls{iid} noise vector $\bm n_k(i)$ to agent~$k$'s local parameter estimate $\bpsi_k(i)$, which is a local linear estimate of $\bmw_k^o$ (cf.\ \cref{projection,Eq:psi'_l,adaptation}), before communicating $\bpsi'_k(i)=\bpsi_k(i)+{\bm n_k}(i)$ to neighboring agents. Let $\sigma_{n,k}^2(i)$ be the variance of each random entry in $\bm n_k(i)$. We call $\bmn_k(i)$ a \emph{privacy noise}. 

Our objective is to find the optimal solution to the following optimization problem:
\begin{subequations}\label{Eq:privacy_req}
\begin{align}
\min_{\sigma_{n,1}^2(i),\ldots,\sigma_{n,N}^2(i)}\ &\ \msd_\net(i) =\frac{1}{N}\sum_{k=1}^N\E\norm{\bmw_k^o-\bmw_{k}(i)}^2,\label{Eq:MSD_net}\\
\st\ &\ \E\norm{\bmw_k^o-\hbmw_{k|\bpsi'_k}(i)}^2\geq \delta_k,\label{Eq:privacy_cost}
\end{align}
\end{subequations}
for $k=1,\ldots,N, i\geq 0$, and where $\hbmw_{k|\bpsi'_k}(i)$ is the least mean-square estimate of $\bmw_k^o$ at time instant $i$ based on $\bpsi'_k(i)$, and $\delta_k\geq 0$ is a privacy threshold chosen according to privacy requirements. Note that the privacy thresholds have to be chosen such that
\begin{align}\label{Ineq:delta_k_range}
0\leq\delta_k \leq \trace{W_{kk}},
\end{align}
otherwise \cref{Eq:privacy_req} is infeasible. 

\begin{remarks}
\item In \cref{Eq:privacy_cost}, it is required that at each time instant $i\geq 0$, the \emph{expected} squared distance $\E\norm{\bmw_k^o-\hbmw_{k|\bpsi'_k}(i)}^2$ over all realizations of $\bmw_k^o$ is no smaller than the predefined parameter $\delta_k$. This provides an inference privacy constraint on the ability of a neighboring agent to agent $k$ in accurately estimating $\bmw_k^o$ based on observation $\bpsi'_k(i)$ on average.

\item In \cref{Eq:privacy_cost}, the estimator $\hbmw_{k|\bpsi'_k}(i)$ is based only on the noisy estimate $\bpsi'_k(i)$. For each neighboring agent $\ell\in\calN_k$, it has access to not only the received noisy estimate $\bpsi'_k(i)$ from agent $k$, but also its own intermediate estimate $\bpsi_\ell(i)$ of $\bmw_\ell^o$. Both of these estimates can be used to infer $\bmw_k^o$ as the unknown parameters $\bmw_k^o$ and $\bmw_\ell^o$ are linearly related with each other through \cref{Eq:cost_constraints}. In this paper, to simplify the analysis and to ensure that the sequence of noise variances $\{\sigma_{n,k}^2(i)\}$ at each agent $k$ is bounded and convergent, we only consider the simplified case in \cref{Eq:privacy_req}. We illustrate the case where a neighboring agent uses its own estimates as additional information using simulations in \cref{Sec:Simulation results}.
\end{remarks}

\section{Privacy-preserving Distributed Projection LMS}\label{Sect: Privacy-preserving Diffusion Strategy}

In this section, we propose an inference privacy mechanism to protect each agent's local task by adding noise to its intermediate estimate before sharing with its neighbors. We then introduce a \emph{weighted} projection operator, which projects neighbors' noisy estimates onto the linear manifold defined by local constraints (\ref{Eq:local_const}), in order to mitigate the negative effect of the additive noises on the estimation accuracy of individual tasks.

\subsection{Adapt-Then-Project Strategy}

In our privacy-preserving distributed projection LMS algorithm, we initialize $\bm w_{k}(-1)=0$ for every agent $k$ in the network. Given data $\{\bm d_k(i),\bm u_{k}(i)\}$ for each time instant $i\geq0$, and for each agent $k=1,\ldots,N$, we perform the following steps iteratively:
\begin{enumerate}
\item Adaptation. Each agent $k$ updates its current estimate $\bm w_{k}(i-1)$ using the stochastic gradient descent (SGD) algorithm to obtain an intermediate estimate
\begin{align}\label{adaptation}
\bm\psi_{k}(i) = \bm w_{k}(i-1) \!+\!\mu_k\bm u_{k}(i)\left(\bm d_k(i)\!-\!\bm u\T_{k}(i)\bm w_{k}(i-1)\right),
\end{align}
where $\mu_k>0$ is the step size at agent $k$.

\item Exchange. Each agent $k$ sends a noisy estimate
\begin{align*}
 \bm\psi'_{k}(i)=\bm\psi_{k}(i) +\bm n_{k}(i),   
\end{align*}
where the random additive noise vector $\bmn_{k}(i)$ is of dimension $M_k\times 1$, to its neighbors and collects estimates $\{\bm\psi'_{\ell}(i)\}$ from neighboring agents $\{\ell\in\calN_k\}$:
\begin{align}\label{Eq:psi'_l}
\bm\psi'_{\ell}(i) =\left\{
\begin{array}{ll}
\bm\psi_{\ell}(i) +\bm n_{\ell}(i),&\textrm{if $\ell\in\calN_k\backslash\{k\}$},\\
\bm\psi_{k}(i),&\textrm{if $\ell=k$}.
\end{array}
\right.
\end{align}

\item Projection. Each agent $k$ projects the estimates $\{\bm\psi'_{\ell}(i)\}_{\ell\in\calN_k}$ received from its neighborhood onto the linear manifold $\{ w_{\calN_k}\mid\calD_{k} w_{\calN_k} + b_k =0\}$ to obtain
\begin{align}\label{projection} 
\bmw_{k}(i) =  [\calP_{\calN_k}(i)]_k \col{\bm\psi'_{\ell}(i)}_{\ell\in\calN_k} - [f_{\calN_k}(i)]_{k},
\end{align}
where the matrix $\calP_{\calN_k}(i)$ and vector $f_{\calN_k}(i)$ are defined in \cref{Eq:f_q,calP}, respectively, in the sequel. 
\end{enumerate}

Let $\delta = \col{\delta_k}_{k=1}^N$, where the non-negative numbers $\{\delta_k\}$ are the privacy thresholds in \cref{Eq:privacy_cost}. We call the proposed scheme \cref{adaptation,Eq:psi'_l,projection} \emph{adapt-then-project} with privacy parameter $\delta$ or ATP$(\delta)$ in short. Specifically, we use ATP(0) algorithm to denote the case where there is no privacy constraint, \emph{i.e.,} no additive noises are introduced in the exchange step in \cref{Eq:psi'_l}. 

\begin{remarks} 
\item The differences between the proposed ATP$(\delta)$ algorithm \cref{adaptation,projection,Eq:psi'_l} and the scheme in \cite{NasRicFer:J17} are in the exchange and projection steps. Specifically, in order to protect each individual task $\bmw_k^o$ against privacy leakage, each agent $k$ sends a noisy intermediate estimate $\bpsi'_k(i)$, instead of the true estimate $\bpsi_k(i)$ as in \cite{NasRicFer:J17}, to its neighboring agents. In addition, in the projection step \cref{projection}, agent $k$ projects its neighboring estimates onto the linear manifold corresponding to the intersection of all $j_k$ linear equality constraints that it is involved in. In contrast, in \cite{NasRicFer:J17},  neighboring estimates of agent $k$ are projected onto the linear manifold corresponding to each of the $j_k$ constraints separately, which generates $j_k$ intermediate estimates, and the new estimate is taken as the average of these intermediate estimates. Moreover, in order to mitigate the negative effect of privacy noises on the projection step, we introduce \emph{weighted} projection operators in \cref{Eq:cost_projection} in the sequel. 
\end{remarks}

To allow a distributed implementation of the privacy mechanism, we make the following assumption.
\begin{Assumption}\label{Ass:additive_noise}(Privacy noise)
The entries of $\bmn_{k}(i)$ at time instant $i$, for any $k=1,\ldots,N$, are \acrshort{iid} with zero mean and variance $\sigma_{n,k}^2(i)$, i.e., 
\begin{align}\label{Eq:R_nk_i}
    R_{n,k}(i) \triangleq \E[\bmn_k(i)\bmn_k\T(i)]=\sigma_{n,k}^2(i)I_{M_k}.
\end{align}
The random noises $\{\bmn_{k}(i)\}$ are white over time and space. The random process $\{\bmn_k(i)\}$ is independent of any other random processes.
\end{Assumption}
From \cref{Ass:additive_noise}, each agent $k$ generates the noise $\bmn_k(i)$ independently of other agents in the network, and also independently over time instants $i\geq0$. Based on this assumption, we now proceed to introduce the weighted projection operator $\calP_{\calN_k}(i)$ and vector $f_{\calN_k}(i)$ involved in \cref{projection}.

\subsection{Weighted Projection Operator}
For each agent $k$, let us collect the noisy intermediate estimates $\{\bpsi'_{\ell}(i)\}$ defined by \cref{Eq:psi'_l} from its neighboring agents $\{\ell\in\calN_k\}$ into an $|\calN_k|\times 1$ block column vector
\begin{align*}
 \bpsi'_{\calN_k}(i) = \col{\bpsi'_{\ell}(i)}_{\ell\in\calN_k}.  
\end{align*}
Then, for each random realization $\bpsi'_{\calN_k}(i)=\psi'_{\calN_k}(i)$, we are interested in seeking the optimal solution to the following optimization problem: 
\begin{subequations}\label{Eq:cost_projection}
\begin{align}
\min_{w_{\calN_k}(i)}\ &\norm{\psi'_{\calN_k}(i)-w_{\calN_k}(i)}^2_{\Omega_{\calN_k}^{-1}(i)}\label{Eq:cost_projection_cost}\\
\st\ &\calD_{k}w_{\calN_k}(i)+b_k=0,
\end{align}
\end{subequations}
where the weight matrix
\begin{align}\label{Eq:weight_matrix}
\Omega_{\calN_k}^{-1}(i) =\diag\left\{\omega_{\ell k }(i)I_{M_\ell}\right\}_{\ell\in\calN_k},
\end{align}
with 
\begin{align}\label{Eq:weights}
  \omega_{\ell k}(i)>0\,\textrm{for $\ell\in\calN_k$, and} \sum_{\ell\in\calN_k}\omega_{\ell k}(i) = 1.  
\end{align}
We note that $\Omega_{\calN_k}^{-1}(i)$ is an $|\calN_k|\times |\calN_k|$ block-diagonal, symmetric and positive definite matrix. We assume that the weights $\{\omega_{\ell k}(i)\}$ defined by \cref{Eq:weight_matrix} converge as $i\to\infty$. Let
\begin{align*}
M_{\calN_k} = \sum_{\ell\in\calN_{k}}M_\ell,
\end{align*}
and
\begin{align}
\calP_{\calN_k}(i)&= I_{M_{\calN_k}} - \Omega_{\calN_k}(i)\calD_k\T(\calD_k \Omega_{\calN_k}(i)\calD_k\T)^{-1}\calD_k \label{calP}\\
f_{\calN_k}(i)&=\Omega_{\calN_k}(i)\calD_k\T(\calD_k\Omega_{\calN_k}(i)\calD_k\T)^{-1}b_k\label{Eq:f_q}.
\end{align}
It follows that matrix $\calP_{\calN_k}(i)$ is a projection matrix. From \cite{Ber:B16}, the minimizer of \cref{Eq:cost_projection} is given by 
\begin{align}\label{Eq:w_k_e_i}
w_{\calN_k}(i)= \calP_{\calN_k}(i)\psi'_{\calN_k}(i) - f_{\calN_k}(i).
\end{align}
For each $\ell\in\calN_k$, we are interested in the projection operator applied to the intermediate estimate from agent $\ell$. This corresponds to an $M_k \times M_\ell$ block of $\calP_{\calN_k}(i)$, denoted as $\brk{\calP_{\calN_k}(i)}_{k,\ell}$. We collect all these in an $N\times N$ block matrix $\mathcal{P}(i)$, whose $M_k\times M_\ell$ $(k,\ell)$-th block equals 
\begin{align}\label{Eq:Pe}
[\mathcal{P}(i)]_{k,\ell} =\left\{\begin{array}{ll}
\brk*{\calP_{\calN_k}(i)}_{k,\ell},& \textrm{if $\ell\in\calN_k$}, \\
0_{M_k\times M_\ell},&\textrm{otherwise}.
\end{array}
\right.
\end{align}
Similarly, we collect the block $\brk{f_{\calN_k}(i)}_k$ corresponding to agent $k$ and define an $N\times 1$ block vector
\begin{align}\label{Eq:f(i)}
    f(i) = \col{\brk{f_{\calN_k}(i)}_k}_{k=1}^N.
\end{align}

\subsection{Weight Matrix}
The matrix $\Omega_{\calN_k}^{-1}(i)$ is introduced in order to mitigate the negative effect of the privacy noises $\{\bmn_k(i)\}$ on the projection step \cref{projection}. In the special case where
\begin{align*}
   \Omega_{\calN_k}^{-1}(i) =  \diag\left\{\frac{1}{|\calN_k|}I_{M_\ell}\right\}_{\ell\in\calN_k},
\end{align*}
the optimal solution of \cref{Eq:cost_projection} is reduced to the Euclidean projection of $\psi'_{\calN_k}(i)$ on the affine set $\braces*{ w_{\calN_k}\mid\calD_{k} w_{\calN_k} + b_k =0}$ \cite[p.398]{BoyVan:B04}. We proceed to rewrite \cref{Eq:cost_projection_cost} as
\begin{align}\label{cost_re}
\norm{\psi'_{\calN_k}(i)-w_{\calN_k}(i)}^2_{\Omega_{\calN_k}^{-1}(i)} = \sum_{\ell\in\calN_k}\omega_{\ell k}(i)\norm{\psi'_{\ell}(i)-w_{\ell}(i)}^2.    
\end{align}
If the privacy noise power $\sigma_{n,\ell}^2(i)$ is large, a small weight $\omega_{\ell k}(i)$ should be assigned to the squared Euclidean distance $\|\psi'_{\ell}(i)-w_{\ell}(i)\|^2$ (see \cref{Eq:omega_lk_i} in \cref{Sec:Simulation results} for a possible choice for the weights $\{\omega_{\ell k}(i)\}$. Thus in the extreme case where $\sigma_{n,\ell}^2(i)\to\infty$ for all $\ell\in\calN_k\backslash\{k\}$, we have $\omega_{\ell k}(i)\to 0$ for all $\ell\in\calN_k\backslash\{k\}$, and $\omega_{k k}(i)\to 1$. Then, from \cref{cost_re}, the optimization problem in \cref{Eq:cost_projection} is reduced to
\begin{subequations}\label{Eq:simplified_projection}
\begin{align}
\min_{w_{\calN_k}(i)}\ &\norm{\psi_{k}(i)-w_{k}(i)}^2\\
\st\ &\calD_{k}w_{\calN_k}(i)+b_k=0.\label{Eq:local constraint}
\end{align}
\end{subequations}
The optimal solution to \cref{Eq:simplified_projection} is then given by
\begin{align*}
    w_k(i)=\psi_k(i),
\end{align*}
and $\{w_\ell(i)\}$ for all $\ell\in\calN_k\backslash\{k\}$ are chosen such that 
\begin{align*}
    \calD_{k}w_{\calN_k}(i)+b_k=0.
\end{align*}
In this case, the proposed ATP($\delta$) algorithm is reduced to the non-cooperative LMS algorithm, namely \cite[p.165]{Say:B08}:
\begin{align}\label{Eq:LMS}
\bmw_{k}(i) = \bm w_{k}(i-1) +\mu_k\bm u_{k}(i)\left(\bm d_k(i)-\bm u\T_{k}(i)\bm w_{k}(i-1)\right).
\end{align}

\subsection{Network Error Dynamics}

To facilitate our analysis in later sections, we derive the network error dynamics in this subsection. Let $M =\sum_{k=1}^N M_k$, and
\begin{align}
\calM &=\diag\{\mu_k I_{M_k}\}_{k=1}^N,\label{Eq:calM}\\
\bmw(i) &= \col{\bmw_{k}(i)}_{k=1}^N,\nn
\bmw^o &= \col{\bm w_{k}^o}_{k=1}^N,\label{Eq:weo}\\
\tbmw(i) &= \bmw^o - \bmw(i),\nn
\bpsi(i) &= \col{\bpsi_k(i)}_{k=1}^N,\label{Eq:psi(i)}\\
\tbpsi(i) &=\bmw^o - \bpsi(i),\label{Eq:tbpsi(i)}\\
\bm\calR_{u}(i) &= \diag\left\{\bm u_{k}(i)\bm u_k\T(i)\right\}_{k=1}^N,\label{Eq:calR_ue_i}\\
\calR_{u} &= \E[\bm\calR_{u}(i)] =\diag\left\{R_{u,k}\right\}_{k=1}^N,\label{Eq:calRue}\\
\bmn(i) &= \col{\bmn_k(i)}_{k=1}^N,\\
R_{n}(i)&= \E[\bmn(i)\bmn\T(i)] = \diag\left\{ R_{n,k}(i)\right\}_{k=1}^N,\label{Eq:R_ne}\\
\bmg(i) &=\col{\bmu_k(i)\bmv_k(i)}_{k=1}^N, \label{Eq:calG_i}\\
\calG &= \E[\bmg(i)\bmg\T(i)] = \diag\left\{R_{u,k}\sigma_{v,k}^2\right\}_{k=1}^N.\label{Eq:calG}
\end{align}
Let $[\calP(i)]_{k^-}$ be the $k$-th block row of $\calP(i)$ defined by \cref{{Eq:Pe}} by setting $[\mathcal{P}(i)]_{k,k}=0_{M_k\times M_k}$, 
\begin{align}
\bmq_{k}(i) &= [\calP(i)]_{k^-}\bmn(i),\label{Eq:q_km}\\
\bmq(i) &= \col{\bmq_{k}(i)}_{k=1}^N.\label{Eq:q_i}
\end{align}
We also define an $N\times N$ block matrix
\begin{align*}
\Gamma(i)= \E[\bmq(i)\bmq\T(i)],
\end{align*}
whose $M_k\times M_\ell$ $(k,\ell)$-th block entry equals
\begin{align}\label{Eq:Gamma_i_km_ln}
[\Gamma(i)]_{k,\ell}&=\E[\bmq_{k}(i)\bmq_{\ell}\T(i)]\nn
&\!\stackrel{\fr{Eq:q_km}}=\E[[\calP(i)]_{k^-}\bmn(i)\bmn\T(i)[\calP(i)]_{\ell^-}\T]\nn
&\!\stackrel{\fr{Eq:R_ne}}=[\calP(i)]_{k^-}R_{n}(i)[\calP(i)]_{\ell^-}\T.
\end{align}
We have the following recursion for the network error vector $\tbmw(i)$.

\begin{Lemma}\label{lem:error_dynamics} 
Consider the distributed strategy \cref{adaptation,projection,Eq:psi'_l}. The evolution of the error dynamics across the network relative to the reference vector $\bmw^o$ defined by (\ref{Eq:weo}) is described by the following recursion:
\begin{dmath}[label={Eq:Netrror_dynamics}]
\tbmw(i) = \calP(i)(I_M-\calM\bm\calR_{u}(i))\tbmw(i-1) -\calP(i)\calM\bmg(i)-\bmq(i),
\end{dmath}
\vspace{-0.2cm}
\begin{dmath}[label={Eq:psierror_dynamics}]
\tbpsi(i+1) =(I_{M}-\calM\bm\calR_{u}(i+1))\calP(i)\tbpsi(i) - (I_{M}-\calM\bm\calR_{u}(i+1))\bmq(i) -\calM\bmg(i+1)
\end{dmath}
for any time instant $i\geq0$.
\end{Lemma}
\begin{proof}
From \cref{adaptation}, we have
\begin{align}\label{Eq:adaptation_net}
\tbpsi(i) = \parens*{I_M-\calM\bm\calR_{u}(i)}\tbmw(i-1) -\calM\bmg(i).
\end{align}
From \cref{projection,Eq:psi'_l}, we obtain
\begin{align}\label{Eq:projection_net}
 \bmw(i) =  \calP(i) \bpsi(i)+\bmq(i)-f(i).   
\end{align}
Since the parameter vectors $\braces*{\bmw_\ell^o}_{\ell\in\calN_k}$ satisfy the local constraints \cref{Eq:local_const} at agent $k$, we have  
\begin{align}
\bmw^o &=  \calP(i) \bmw^o-f(i).\label{Eq:projection_unknown}
\end{align}
Subtracting \cref{Eq:projection_net} from both sides of \cref{Eq:projection_unknown}, we have
\begin{align}\label{Eq:projection_net_error}
 \tbmw(i) =  \calP(i) \tbpsi(i)-\bmq(i).   
\end{align}
Substituting \cref{Eq:adaptation_net} into the \gls{RHS} of \cref{Eq:projection_net_error}, we arrive at the desired recursion (\ref{Eq:Netrror_dynamics}). Substituting \cref{Eq:projection_net_error} into the \gls{RHS} of \cref{Eq:adaptation_net}, we obtain the desired recursion (\ref{Eq:psierror_dynamics}), and the proof is complete.
\end{proof}

Let $\Sigma$ be a symmetric positive semi-definite matrix, and
\begin{align}\label{Eq:Sigma'}
\Sigma'(i)=\E[(I_{M}-\bm\calR_{u}(i)\calM)\calP\T(i)\Sigma\calP(i)(I_{M}-\calM\bm\calR_{u}(i))].
\end{align}
From (\ref{Eq:Netrror_dynamics}), we have
\begin{align}\label{Eq:weighted_norm}
\E\norm{\tbmw(i)}^2_\Sigma&
\stackrel{\ft{(a)}}=\E\norm{\tbmw(i-1)}^2_{\Sigma'(i)} +\E[\bmq\T(i)\Sigma\bmq(i)]\nonumber\\
&\hspace{0.4cm}+\E[\bmg\T(i)\calM\calP\T(i)\Sigma\calP(i)\calM\bmg(i)]\nonumber\\
&\stackrel{\ft{(b)}}=\E\norm{\tbmw(i-1)}^2_{\Sigma'(i)}
 +\E[\trace{\bmq(i)\bmq\T(i)\Sigma}]\nonumber\\
 &\hspace{0.4cm}+\E[\trace{\bmg(i)\bmg\T(i)\calM\calP\T(i)\Sigma\calP(i)\calM}]\nn
&\stackrel{\ft{(c)}}=\E\norm{\tbmw(i-1)}^2_{\Sigma'(i)} +\trace{\Gamma(i)\Sigma}\nn
&\hspace{0.4cm}+\trace{\calG\calM\calP\T(i)\Sigma\calP(i)\calM},
\end{align}
where in (a) we used \cref{Ass:measurement_noise,Ass:Regression_data,Ass:Independence_assumption,Ass:additive_noise}, and 
\begin{align}
    \E\bmg(i) &= 0_{M\times1},\label{Eq:Eg(i)}\\
    \E\bmq(i) &= 0_{M\times1},\label{Eq:Eq(i)}
\end{align}
in (b) we used the property $\trace{AB}=\trace{BA}$ for any matrices $\braces*{A,B}$ of compatible sizes, and in (c) we interchanged expectation and trace.

\section{Privacy Noise Design}\label{Sect: Privacy Mechanism Noise Design and Convergence Analysis}
In this section, we present an approximate solution to the utility-privacy optimization trade-off problem \cref{Eq:privacy_req} by first deriving a sufficient condition that leads to the privacy constraints. Finally, we motivate a distributed and adaptive scheme for each agent to compute its local privacy noise power at each time instant.

\subsection{Privacy Noise Power}
We start by showing that the quantity $\msd_\net(i)$ is a monotonically increasing function \gls{wrt} the privacy noise powers $\left\{\sigma_{n,k}^2(i)\right\}$. 
\begin{Lemma}\label{Thm: MSD vs privacy noise powers}
Consider the distributed strategy \cref{adaptation,projection,Eq:psi'_l}. Suppose that \cref{Ass:measurement_noise,Ass:Regression_data,Ass:Independence_assumption,Ass:additive_noise} hold. Then, $\msd_\net(i)$ is an increasing function \gls{wrt} the privacy noise powers $\left\{\sigma_{n,k}^2(i)\right\}$.
\end{Lemma}
\begin{proof}
From \cref{Eq:MSD_net}, we have
\begin{align}\label{Eq:transient_MSD}
    \msd_\net(i)=\E\norm{\tbmw(i)}^2_{I_{M}/N}.
\end{align}
By setting 
\begin{align}\label{Eq:Sigma}
\Sigma=\frac{1}{N}I_M    
\end{align}
on both sides of \cref{Eq:weighted_norm}, we observe that the error $\tbmw(i-1)$ on the \gls{RHS} of \cref{Eq:weighted_norm} only relies on the random variables $\{\bmd_k(j),\bmu_k(j), \bmn_k(j) \mid k=1,\ldots,N, 0\leq j \leq i-1\}$, thus its value is not affected by the variances $\{\sigma_{n,k}^2(i)\}$ at time instant $i$. The value of $\msd_\net(i)$ depends on the variances $\{\sigma_{n,k}^2(i)\}$ via the matrix $\Gamma(i)$ in \cref{Eq:Gamma_i_km_ln}. We have
\begin{align}\label{Eq:trace_Gamma_Sigma}
\trace{\Gamma(i)\Sigma} &\stackrel{\fr{Eq:Sigma}}=  \frac{1}{N}\trace{\Gamma(i)}\nn
&=\frac{1}{N}\sum_{k=1}^{N}\trace{\brk*{\Gamma(i)}_{k,k}}  
\end{align}
and
\begin{align}\label{Eq:trace_Gamma_kk}
    \trace{\brk*{\Gamma(i)}_{k,k}} &\stackrel{\fr{Eq:Gamma_i_km_ln}}= \trace{\sum_{\ell\in\calN_k\backslash\braces{k}}\brk*{\calP(i)}_{k,\ell}R_{n,\ell}\,\brk*{\calP(i)}_{k,\ell}\T}\nn
    &=\sum_{\ell\in\calN_k\backslash\braces{k}}\trace{\brk*{\calP(i)}_{k,\ell}R_{n,\ell}\,\brk*{\calP(i)}_{k,\ell}\T}\nn
    &=\sum_{\ell\in\calN_k\backslash\braces{k}}\trace{R_{n,\ell}\,\brk*{\calP(i)}_{k,\ell}\T\brk*{\calP(i)}_{k,\ell}}\nn
    &\stackrel{\fr{Eq:R_nk_i}}=\sum_{\ell\in\calN_k\backslash\braces{k}}\sigma_{n,\ell}^2(i)\trace{\brk*{\calP(i)}_{k,\ell}\T\brk*{\calP(i)}_{k,\ell}},
\end{align}
which completes the proof.
\end{proof}

In view of \cref{Eq:privacy_req}, we need to find the smallest privacy noise powers that satisfy the privacy constraints \cref{Eq:privacy_cost}.
For each agent $k=1,\ldots,N$, let
\begin{align}
U_{kk}(i) &=\E[(\bmw^o_k-\E\bmw^o_k)(\bpsi'_k(i)-\E\bpsi'_k(i))\T],\label{Eq:U_kk_i}\\
X_{kk}(i) &= \E[\left(\bpsi_k(i)-\E\bpsi_k(i)\right)\left(\bpsi_k(i)-\E\bpsi_k(i)\right)\T],\label{Eq:Xkk(i)}\\
X'_{kk}(i) &= \E[\left(\bpsi'_k(i)-\E\bpsi'_k(i)\right)\left(\bpsi'_k(i)-\E\bpsi'_k(i)\right)\T]\label{Eq:R_psi'_k_i},
\end{align}
and
\begin{align*}
\hbmw_{k|\bpsi'_{k}}(i) = U_{kk}(i)\left(X'_{kk}(i)\right)^{-1}\left(\bpsi'_k(i)-\E\bpsi'_k(i)\right)+\E\bmw_k^o
\end{align*}
be the linear least-mean-square estimator \cite[p.66]{Say:B08} of $\bmw_k^o$ at time instant $i$, given $\bpsi'_{k}(i)$ (note that $\hbmw_{k|\bpsi'_{k}}(i)$ and $\bpsi'_k(i)$ are linearly related in our ATP$(\delta)$ strategy).


From \cref{Ass:additive_noise}, we have
\begin{align}\label{relation}
X'_{kk}(i) = X_{kk}(i) + R_{n,k}(i)=X_{kk}(i) + \sigma_{n,k}^2(i)I_{M_k}
\end{align}
and using \cite[p.66]{Say:B08}, we obtain
\begin{align}\label{Eq:MSD_llmse}
&\hspace{-0.5cm}\E\norm{\bmw_k^o-{\hbmw}_{k|\psi'_k}(i)}^2\nn
&=\trace{W_{kk} -U_{kk}(i)\left(X'_{kk}(i)\right)^{-1}U\T_{kk}(i)}\nn
&=\trace{W_{kk}} -\trace{U_{kk}(i)\left(X'_{kk}(i)\right)^{-1}U\T_{kk}(i)},
\end{align}
where the matrix $W_{kk}$ is defined in \cref{Wkk}. Substituting \cref{Eq:MSD_llmse} into the \gls{LHS} of the privacy constraint (\ref{Eq:privacy_cost}), we have
\begin{multline}\label{rewrittenCost}
\trace{U_{kk}(i)\left(X_{kk}(i) + \sigma_{n,k}^2(i)I_{M_k}\right)^{-1}U_{kk}\T(i)}\\ \leq \trace{W_{kk}}-\delta_k.
\end{multline}
A numerical solution for the optimal noise powers $\{\sigma_{n,k}^2(i)\}$ can be obtained by seeking the smallest values that satisfy the inequality \cref{rewrittenCost}. However, in view of the fact that agents in the network work in a cooperative manner, the evaluation of $\braces*{U_{kk}(i),X_{kk}(i)}$ involves global statistics that are not available locally at each agent $k$. In this paper, we are interested in a \emph{distributed and adaptive} scheme to compute the noise power $\sigma_{n,k}^2(i)$ for each agent $k$ and time instant $i\geq 0$. To this end, we first derive a closed-form expression for a noise power $\hsigma_{n,k}^2(i)$ that satisfies the privacy constraint \cref{Eq:privacy_cost} for each agent $k$ and time instant $i\geq 0$.

\begin{Theorem}\label{Thm:Sufficient condition}
Consider the distributed strategy \cref{adaptation,projection,Eq:psi'_l}. Suppose that \cref{Ass:measurement_noise,Ass:additive_noise,Ass:Independence_assumption,Ass:Regression_data} hold. If
\begin{align}\label{Ineq:sufficientCondition}
\sigma_{n,k}^2(i)\geq \frac{\trace{ U_{kk}\T(i) U_{kk}(i) }}{\trace{W_{kk} }-\delta_k}
\end{align}
for each agent $k$ and time instant $i\geq 0$, then the privacy constraint (\ref{Eq:privacy_cost}) is satisfied. 
\end{Theorem}
\begin{proof}
The proof involves mainly algebraic manipulations and is provided in \cref{Sup:sufficient condition} of the supplementary material.
\end{proof} 

In view of \cref{Thm: MSD vs privacy noise powers}, to reduce $\msd_\net(i)$ as much as possible, we set the noise power of agent $k$ at time instant $i\geq 0$ to be
\begin{align}\label{Eq:sigma_nk_2}
\widehat{\sigma}_{n,k}^2(i)= \frac{\trace{ U_{kk}\T(i) U_{kk}(i) }}{\trace{W_{kk} }-\delta_k}.
\end{align}
We note that this is a feasible solution to \cref{Eq:privacy_req} due to \cref{Thm:Sufficient condition} and moreover $\widehat{\sigma}_{n,k}^2(i)>0$ because $\trace{U_{kk}\T(i)U_{kk}(i)}> 0$ and \cref{Ineq:delta_k_range}. We next show that the sequence $\braces*{\widehat{\sigma}_{n,k}^2(i)}$ defined by \cref{Eq:sigma_nk_2} is convergent. To facilitate our analysis, we make the following assumption.
\begin{Assumption}\label{Ass:Collumn_rank_def}
The matrix $\brk*{\calD_{k_o}}_{\cdot,k_o}$ is column-rank deficient for at least one agent $k_o$.
\end{Assumption}
This assumption definitely holds true for sufficiently large enough $M_{k_o}$ such that $\rank\parens*{\calD_{k_o}}<M_{k_o}$ holds.

\begin{Lemma}
Consider the matrix $\calP(i)$ defined by \cref{Eq:Pe}. Suppose that \cref{Ass:Collumn_rank_def} holds. Suppose also that the matrices $\braces*{\Omega_{\calN_k}(i)}$ are positive-definite for all agents $k$ and time instants $i\geq0$. For any time instant $i\geq0$, we have
\begin{align}\label{Ineq:norm_CalP(i)}
    \norm{\calP(i)}\geq 1.
\end{align}
\end{Lemma}
\begin{proof}
Let $\psi$ be an $N\times1$ block vector, with each block entry of size $M_k\times1$. We choose $\psi=\brk*{0\T,\ldots,0\T,\psi_{k_o}\T,0\T,\ldots,0\T}\T$ with a non-zero vector at the $k_o$-th block for some $k_o=1,\ldots,N$. We have
\begin{align}\label{Ineq:norm_P_psi}
    \norm{\calP(i)\psi}^2&=\sum_{k=1}^N\norm{\brk*{\calP(i)}_k\psi}^2\nn
    &\geq\norm{\brk*{\calP(i)}_{k_o}\psi}^2\nn
    &=\norm{\brk*{\calP(i)}_{k_o,k_o}\psi_{k_o}}^2.
\end{align}
From \cref{calP,Eq:Pe}, we obtain
\begin{align*}
    &\hspace{-0.1cm}\brk*{\calP(i)}_{k_o,k_o}\nn
    &=\brk*{\calP_{\calN_{k_o}}(i)}_{k_o,k_o}\nn
    &= \brk*{I_{M_{\calN_{k_o}}} - \Omega_{\calN_{k_o}}(i)\calD_{k_o}\T(\calD_{k_o} \Omega_{\calN_{k_o}}(i)\calD_{k_o}\T)^{-1}\calD_{k_o}}_{k_o,k_o}\nn
    &=I_{M_{k_o}}-\brk*{\Omega_{\calN_{k_o}}(i)\calD_{k_o}\T}_{k_o}(\calD_{k_o} \Omega_{\calN_{k_o}}(i)\calD_{k_o}\T)^{-1}\brk*{\calD_{k_o}}_{\cdot,k_o}\nn
    &\stackrel{\fr{Eq:weight_matrix}}=I_{M_{k_o}}-\frac{1}{\omega_{k_o,k_o}(i)}\brk*{\calD_{k_o}\T}_{k_o}(\calD_{k_o} \Omega_{\calN_{k_o}}(i)\calD_{k_o}\T)^{-1}\brk*{\calD_{k_o}}_{\cdot,k_o}.
\end{align*}
Note that the matrix $\brk*{\calD_{k_o}\T}_{k_o}(\calD_{k_o} \Omega_{\calN_{k_o}}(i)\calD_{k_o}\T)^{-1}\brk*{\calD_{k_o}}_{\cdot,k_o}$ is positive semi-definite, and it has a zero eigenvalue given \cref{Ass:Collumn_rank_def}. Therefore $\brk*{\calP(i)}_{k_o,k_o}$ has an eigenvalue 1. In view of the fact that the $\brk*{\calP(i)}_{k_o,k_o}$ is a symmetric matrix, we conclude that 
\begin{align*}
    \norm{\brk*{\calP(i)}_{k_o,k_o}}\geq1.
\end{align*}
Thus there exists $\psi_{k_o}$ such that
\begin{align}\label{Ineq:norm_P_kokoPsi_ko}
    \norm{\brk*{\calP(i)}_{k_o,k_o}\psi_{k_o}}\geq\norm{\psi_{k_o}}=\norm{\psi}.
\end{align}
It then follows from \cref{Ineq:norm_P_kokoPsi_ko,Ineq:norm_P_psi} that
\begin{align*}
    \norm{\psi}\leq\norm{\calP(i)\psi}\leq\norm{\calP(i)}\norm{\psi},
\end{align*}
and the proof is complete.
\end{proof}

\begin{Proposition}\label{Lem:steady-state variance} 
Consider the distributed strategy \cref{adaptation,projection,Eq:psi'_l}. Suppose that \cref{Ass:Collumn_rank_def,Ass:measurement_noise,Ass:additive_noise,Ass:Independence_assumption,Ass:Regression_data} hold. Suppose also that 
\begin{align*}
\frac{1-1/\norm{\calP(i)}}{1+1/\norm{\calP(i)}}<\frac{\lambda_{\min}(R_{u,k})}{\lambda_{\max}(R_{u,k})}\leq 1
\end{align*}
for all agents $k$. If the step-size $\mu_k$ satisfies
\begin{align}\label{Cond:variance convergence}
\frac{1-1/\norm{\calP(i)}}{\lambda_{\min}(R_{u,k})}<\mu_k<\frac{1+1/\norm{\calP(i)}}{\lambda_{\max}(R_{u,k})}
\end{align}
for each agent $k$, then
\begin{align}
\label{Eq:lim_sigma_n_k_i}
\lim_{i\to\infty}\widehat{\sigma}_{n,k}^2(i) =\frac{\trace{ W_{kk}^2  }}{\trace{W_{kk} }-\delta_k},
\end{align}
where $\widehat{\sigma}_{n,k}^2(i)$ is defined by \cref{Eq:sigma_nk_2}.
\end{Proposition}
\begin{proof}
See \cref{Sup:proof of steady-state variance} of the supplementary material.
\end{proof}

We next motivate a scheme based on \cref{Eq:lim_sigma_n_k_i} for agents in the network to compute the privacy noise powers in a distributed and adaptive manner.

\subsection{Distributed and Adaptive \texorpdfstring{ATP($\delta$)}{ATP}}
The covariance matrix $W_{kk}$ may not be known beforehand, and may change over time. In this subsection, we motivate a distributed and adaptive scheme that enables each agent $k$ to estimate $\widehat{\sigma}_{n,k}^2(\infty)$ defined by \cref{Eq:lim_sigma_n_k_i} locally by using data available to it at time instant $i$. We assume that the statistics $\E\bmw_k^o$ is available at each agent $k$ \emph{a priori}; otherwise it can be estimated via 
\begin{align*}
    \E\bmw_k^o=R_{u,k}^{-1}\,\E[\bmd_k(i)\bmu_k(i)]
\end{align*}
which follows from \cref{model,Ass:Independence_assumption,Ass:Regression_data}, and where the quantities $\braces*{R_{u,k},\E[\bmd_k(i)\bmu_k(i)]}$ can be inferred from the random data $\braces*{\bmd_k(i),\bmu_k(i)}$. Let $\bbeta_k(i)$ be an estimate of $\trace{ W_{kk}^2 }$, and $\bgamma_k(i)$ an estimate of $\trace{W_{kk} }$ respectively at time instant $i$. We start with $\bbeta_{k}(-1)= 0$ and $\bgamma_k(-1)=\delta_k$ for every agent $k$. For any time instant $i\geq 0$, we compute the quantities $\braces*{\bbeta_k(i),\bgamma_k(i)}$ at each agent $k$ using
\begin{align}
\bmR_{\psi,k}(i) &= \parens*{\bpsi_k(i)-\E\bmw_k^o}\parens*{\bpsi_k(i)-\E\bmw_k^o}\T,\nn
\bbeta_k(i) &=\alpha\bbeta_k(i-1)+\parens*{1-\alpha}\trace{\bmR_{\psi,k}^2(i)},\nn
\bgamma_k(i) &=\alpha\bgamma_k(i-1)+\parens*{1-\alpha}\trace{\bmR_{\psi,k}(i)},\nonumber
\end{align}
where $\bpsi_k(i)$ is the intermediate estimate generated by \cref{adaptation}, and the parameter $0<\alpha<1$ is a forgetting factor. Then, agent $k$ estimates $\widehat{\sigma}_{n,k}^2(\infty)$ using
\begin{align}\label{Eq:adaptive_privacy_noise_power}
 \widehat{\widehat{\bsigma}}^2_{n,k}(i) =\left\{\begin{array}{lr}
    \alpha\widehat{\widehat{\bsigma}}^2_{n,k}(i-1)+\parens*{1-\alpha}\frac{\bbeta_k(i)}{\bgamma_k(i)-\delta_k},  \\ &\hspace{-2.5cm}\textrm{if ${\bbeta_k(i)}/\parens*{\bgamma_k(i)-\delta_k}>0$},  \\
    \widehat{\widehat{\bsigma}}^2_{n,k}(i-1),  & \textrm{otherwise},
 \end{array} \right.  
\end{align}
where $\widehat{\widehat{\bsigma}}^2_{n,k}(-1)=0$.
A privacy noise $\bmn_{k}(i)$ is generated at agent $k$ by following a distribution with zero mean and variance $\widehat{\widehat{\bsigma}}^2_{n,k}(i)$. This noise is added to the intermediate estimate $\bpsi_k(i)$ to form the noisy estimate $\bpsi'_k(i)$. The quantity $\bpsi'_k(i)$ is then transmitted to the neighboring agents $\braces*{\ell\in\calN_k\backslash\braces*{k}}$. \cref{Fig:Algorithm} summarizes our proposed distributed and adaptive ATP($\delta$) algorithm.
\begin{figure}[!htb]
\centering
\includegraphics[width=3.45in]{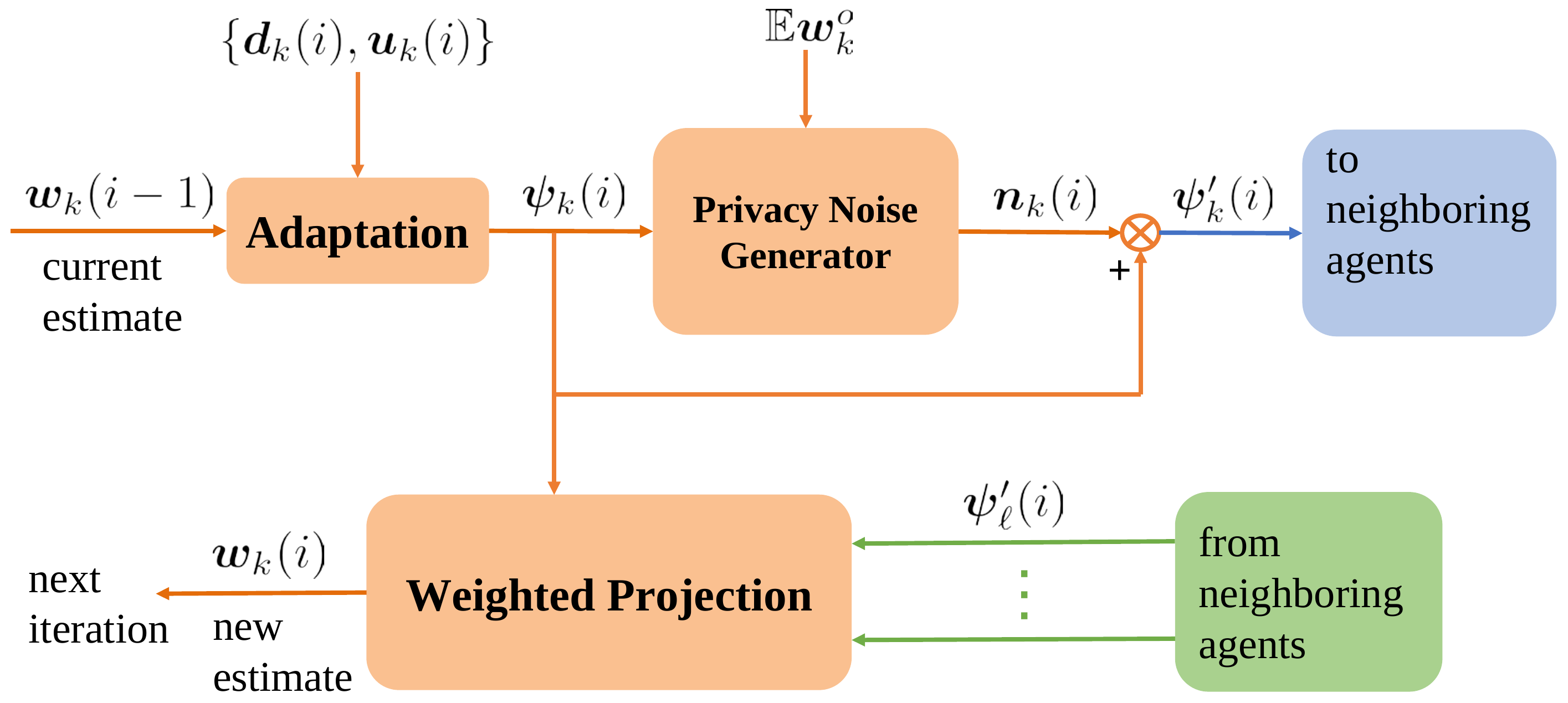}
\caption{Diagram of the proposed distributed and adaptive ATP$(\delta)$ algorithm.} 
\label{Fig:Algorithm}
\end{figure}

\section{Performance Analysis}\label{Sect: Performance Analysis}
In this section, we study the mean and mean-square behaviors of the proposed ATP($\delta$) algorithm with privacy noise powers $\braces*{\sigma_{n,k}^2(i)=\widehat{\sigma}_{n,k}^2(i)}$, where the quantities $\braces*{\widehat{\sigma}_{n,k}^2(i)}$ are defined by \cref{Eq:sigma_nk_2}. In addition, we evaluate its privacy-preserving performance.

\subsection{Mean Behavior Analysis}
We have the following theorem for the network mean performance.

\begin{Theorem}\label{Thm:Network mean performance}
Consider the distributed strategy \cref{adaptation,projection,Eq:psi'_l} with privacy noise powers $\braces*{\sigma_{n,k}^2(i)=\widehat{\sigma}_{n,k}^2(i)}$. Suppose that \cref{Ass:Collumn_rank_def,Ass:Independence_assumption,Ass:measurement_noise,Ass:additive_noise,Ass:Regression_data} hold, and 
\begin{align*}
\frac{1-1/\norm{\calP(i)}}{1+1/\norm{\calP(i)}}<\frac{\lambda_{\min}(R_{u,k})}{\lambda_{\max}(R_{u,k})}\leq 1
\end{align*}
for all agents $k=1,\ldots,N$. If the step-size $\mu_k$ is chosen to satisfy
\begin{align}\label{Cond:mean stability}
\frac{1-1/\norm{\calP(i)}}{\lambda_{\min}(R_{u,k})}<\mu_k<\frac{1+1/\norm{\calP(i)}}{\lambda_{\max}(R_{u,k})}
\end{align}
for each agent $k$, then
\begin{align}\label{Eq:mean_net_err}
\lim_{i\to\infty}\E\tbmw(i) = 0.
\end{align}
\end{Theorem}
\begin{proof}
Let $\calA\parens{i}\triangleq\calP(i)\parens*{I_M-\calM\calR_{u}}$. Taking expectations on both sides of (\ref{Eq:Netrror_dynamics}) gives
\begin{align}\label{Eq:recursion_E_tbmw}
\E\tbmw(i) &= \calP(i)\E[(I_M-\calM\bm\calR_{u}(i))\tbmw(i-1)] -\calP(i)\calM\E\bmg(i)\nn
&\hspace{0.4cm}-\E\bmq(i)\nn
&\stackrel{\ft{(a)}}=\calA(i)\E\tbmw(i-1) -\calP(i)\calM\E\bmg(i)-\E\bmq(i)\nn
&\stackrel{\ft{(b)}}=\calA(i)\E\tbmw(i-1)
\end{align}
where in step (a) we used \cref{Ass:Regression_data}, and step (b) follows from \cref{Ass:Independence_assumption,Ass:measurement_noise,Ass:additive_noise}. Using the same proof as in \cref{lem:calV} of the supplementary material, $\rho(\calA\parens{i})<1$ if \cref{Cond:mean stability} holds. The proof is now complete.
\end{proof}

\subsection{Mean-square Performance Analysis}\label{Sect:Mean-square Performance Analysis}
In this subsection, we evaluate the steady-state network MSD, namely, the value of $\msd_\net(i)$ as the time instant $i\to\infty$. Let
\begin{multline}\label{Eq:calF}
\calF(i)=\E\,\big[\left((I_{M}-\bm\calR_{u}(i)\calM)\calP\T(i)\right)\\
\otimes\left((I_{M}-\bm\calR_{u}(i)\calM)\calP\T(i)\right)\big],
\end{multline}
where $\otimes$ denotes Kronecker product.

\begin{Assumption}\label{assumpt:sigma_conv}
The noise powers $\{\sigma_{n,k}^2(i)\}$ converge as $i\to\infty$ for all $k=1,\ldots,N$.
\end{Assumption}

We note that \cref{assumpt:sigma_conv} is satisfied if we set $\sigma_{n,k}^2(i)=\widehat{\sigma}_{n,k}^2(i)$ for all agents $k=1,\ldots,N$ and $i\geq0$. \cref{assumpt:sigma_conv} together with the assumption that $\omega_{\ell k}(i)$ defined by \cref{Eq:weights} converges as $i\to\infty$ for any agents $\ell,k=1,\cdots,N$ lead to the following conclusions:
\begin{enumerate}[(a)]
	\item $\Gamma(i)$ converges as $i\to\infty$.
	\item $\calF=\lim_{i\to\infty}\calF(i)$ exists.
\end{enumerate}
In addition, we assume that the quantity $\parens*{1-1/\norm{\calP(i)}}\big/{\lambda_{\min}(R_{u,k})}$ on the \acrshort{LHS} of \cref{Cond:mean stability} is sufficiently small, so that for sufficiently small step-sizes $\braces*{\mu_k}$, we have
\begin{align*}
    \calF=\calA\T(\infty)\otimes\calA\T(\infty) + O(\mu_\max^2)
\end{align*}
where $\mu_\max\triangleq\max\braces*{\mu_1,\ldots,\mu_N}$. Then if the step-sizes $\braces*{\mu_k}$ satisfy \cref{Cond:mean stability}, we have $\rho\parens*{\calA(\infty)}<1$, and $\rho\parens*{\calF}\thickapprox\parens*{\rho\parens*{\calA(\infty)}}^2<1$, i.e., the matrix $\calF$ is stable. Let 
\begin{align}\label{Eq:sigma_ss}
\sigma_{\textrm{ss}} = \frac{1}{N}\left(I_{M^2}-\calF\right)^{-1}\vect\left(I_{M}\right).
\end{align}

\begin{Theorem}\label{Thm:Network MSD performance}
Consider the distributed strategy \cref{adaptation,projection,Eq:psi'_l}. Suppose that \cref{Ass:Collumn_rank_def,Ass:measurement_noise,Ass:Regression_data,Ass:Independence_assumption,Ass:additive_noise,assumpt:sigma_conv} hold. Then,
\begin{align}\label{Eq:msd_net}
\msd_\net &\triangleq \lim_{i\to\infty}\msd_\net(i)\nn &=\lim_{i\to\infty}\left[\vect\left(\calP(i)\calM\calG\calM\calP\T(i)\right)\right]\T\sigma_{\textrm{ss}}\nn
&\hspace{0.4cm}+\lim_{i\to\infty}\left[\vect\left(\Gamma(i)\right)\right]\T\sigma_{\textrm{ss}}.
\end{align}
\end{Theorem}
\begin{proof}
Note that \cite[p.762]{Say:J14}
\begin{align}
\trace{AB}&=\left[\vect(B\T)\right]\T\vect(A)\label{Eq:vect_trace}\\
\vect(ACB)&=(B\T\otimes A)\vect(C)\label{Eq:vect_product}
\end{align}
for any matrices $\{A,B,C\}$ of compatible sizes. Let $\sigma = \vect(\Sigma)$. Then, it follows from \cref{Eq:vect_product} that
\begin{align*}
\sigma'(i)\triangleq\vect\left(\Sigma'(i)\right)
=\calF(i)\sigma.
\end{align*}
Now, we rewrite the recursion \cref{Eq:weighted_norm} as
\begin{align}\label{Eq:weighted_norm_vec}
\E\norm{\tbmw(i)}^2_\sigma &= \E\norm{\tbmw(i-1)}^2_{\calF(i)\sigma}+\left[\vect\left(\Gamma(i)\right)\right]\T\sigma\nn
&\hspace{0.4cm}+\left[\vect\left(\calP(i)\calM\calG\calM\calP\T(i)\right)\right]\T\sigma,
\end{align}
where we used the property \cref{Eq:vect_trace}. From \cref{assumpt:sigma_conv}, we obtain
\begin{align*}
\lim_{i\to\infty}\E\norm{\tbmw(i)}^2_{\left(I_{M^2}-\calF\right)\sigma} &= \lim_{i\to\infty}\left[\vect\left(\calP(i)\calM\calG\calM\calP\T(i)\right)\right]\T\sigma\nn
&\hspace{0.4cm}+\lim_{i\to\infty}\left[\vect\left(\Gamma(i)\right)\right]\T\sigma.
\end{align*}
By setting $\sigma=\sigma_{\textrm{ss}}$, which is defined by \cref{Eq:sigma_ss} on the both sides, we arrive at the desired result \cref{Eq:msd_net} and the proof is complete.
\end{proof}

\subsection{Privacy-preserving Performance}\label{Sect: Performance Metrics of Inference Privacy Preservation}
In this subsection, we introduce the performance metrics for inference privacy preservation, and evaluate the privacy-preserving performance of the proposed ATP$(\delta)$ algorithm.

Let $\bPhi_\ell(i)$ be the set of information that an agent $\ell$ has at time instant $i\geq0$ to infer its neighboring agent $k$'s task $\bmw_k^o$. To quantify the privacy-preserving performance of any scheme, we evaluate, in terms of mean-square error, how well agent $\ell$ can estimate $\bmw_k^o$. Then we average over the whole network. The network inference privacy error is defined using
\begin{align}
\xi_{\textrm{net}}(i)&=\frac{1}{N}\sum_{k=1}^N\frac{1}{|\calN_k|-1}\sum_{\ell\in\calN_k\backslash\{k\}}\E\norm{ \bmw_k^o-\hbmw_{k\mid\bPhi_\ell}(i) }^2,\label{Eq:inference_privacy}
\end{align}
where $\hbmw_{k\mid\bPhi_\ell}(i)$ is the least mean-square estimate of $\bmw_k^o$ given $\bPhi_\ell(i)$. A larger $\xi_{\textrm{net}}(i)$ implies better average privacy preservation across the network at time instant $i$.

In our numerical experiments, we compare the performance of our proposed ATP$(\delta)$ algorithm with the proposed ATP(0) algorithm, the multitask diffusion algorithm \cite{NasRicFer:J17} (denoted as MDA for convenience) and the non-cooperative LMS algorithm \cref{Eq:LMS} (denoted as NoCoop). Since the unknown coefficient vectors $\braces*{\bmw_k^o}$ are linearly related across the network, agent $\ell$ can make use of all neighboring estimates $\braces*{\bpsi'_s(i)}$ that are available to it at time instant $i$ to infer its neighbor's task $\bmw_k^o$ for any $k\in\calN_\ell$. However, for simplicity, in order to evaluate the privacy-preserving performance of the different schemes, we constrain the information set $\bPhi_\ell(i)$ as follows:
\begin{itemize}
	\item For ATP$(\delta)$, we let $\bPhi_\ell(i)=\{\bpsi_\ell(i),\bpsi'_k(i)\}$, i.e., each agent $\ell$ is memoryless across time and spatial domains.
	\item For ATP$(0)$ and MDA, we let $\bPhi_\ell(i)=\{\bpsi_\ell(i),\bpsi_k(i)\}$ since no privacy noises are added in these schemes.
	\item For NoCoop, we let $\bPhi_\ell(i)=\{\bmw_\ell(i)\}$ since there is no information exchange between agents. 
\end{itemize}

Let
\begin{gather}
\calW =\E[(\bmw^o-\E\bmw^o)(\bmw^o-\E\bmw^o)\T],\label{Eq:calW}\\
\Psi(i) =\E[\bpsi(i)\bpsi\T(i)]\label{Eq:Psi(i)},\\
\calV(i)=\E[(\bmw^o-\E\bmw^o)(\tbpsi(i)-\E\tbpsi(i))\T],\label{Eq:calV(i)}\\
\bmr_{du}(i) = \col{ \bmu_k(i)\bmd_k(i)}_{k=1}^N=\bm\calR_{u}(i)\bmw^o+\bmg(i),\label{Eq:r_du}\\
\calB(i) = (I_{M}-\calM\calR_{u})\calP(i).\label{Eq:calB(i)}
\end{gather} 
Let $\bpsi_{\braces*{k,\ell}}(i)=\brk*{\bpsi_k\T(i),\bpsi_\ell\T(i)}\T$. We also define
\begin{align}\label{Eq:U_k_kl}
    U_{k\braces*{k,\ell}}(i) =\E[\parens*{\bmw_k^o-\E\bmw_k^o}\parens*{\bpsi_{\braces*{k,\ell}}(i)-\E\bpsi_{\braces*{k,\ell}}(i)}\T],
\end{align}
 \vspace{-0.5cm}
\begin{align}\label{Eq:X_kl_kl}    
    X_{{\braces*{k,\ell}}{\braces*{k,\ell}}}(i) &=\E\Big[\parens*{\bpsi_{\braces*{k,\ell}}(i)-\E\bpsi_{\braces*{k,\ell}}(i)}\nn
    &\hspace{0.7cm}\times\parens*{\bpsi_{\braces*{k,\ell}}(i)-\E\bpsi_{\braces*{k,\ell}}(i)}\T\Big],
\end{align}
and 
\begin{align*}
    \widecheck{R}_{n,k}(i)&=\left[\begin{array}{cc}
     R_{n,k}(i)    &0  \\
        0 & 0
    \end{array}\right]
\end{align*}
where the quantity $R_{n,k}(i)$ is defined by \cref{Eq:R_nk_i}. Let
\begin{multline}
 \calH(i) = \E\,\big[\left(\left(I_{M}-\calM\bm\calR_{u}(i+1)\right) \calP(i)\right)\\
 \otimes\left(\left(I_{M}-\calM\bm\calR_{u}(i+1)\right) \calP(i)\right)\big],\label{Eq:H(i)}  
\end{multline}
\vspace{-0.8cm}
\begin{multline}
 \calX(i)=\E\,\big[\left(\left(I_{M}-\calM\bm\calR_{u}(i+1)\right) f(i)\right)\\
 \otimes\left(\left(I_{M}-\calM\bm\calR_{u}(i+1)\right)\calP(i)\right)\big],\label{Eq:X(i)}  
\end{multline}
\vspace{-0.8cm}
\begin{multline}
\calX'(i)=\E\,\big[\left(\left(I_{M}-\calM\bm\calR_{u}(i+1)\right)\calP(i)\right)\\
\otimes\left(\left(I_{M}-\calM\bm\calR_{u}(i+1)\right) f(i)\right)\big],\label{Eq:X'(i)}    
\end{multline}
\begin{align}
\calY(i)&=\E[\left(\calM\bmr_{du}(i+1)\right)\otimes\left(\left(I_{M}-\calM\bm\calR_{u}(i+1)\right) \calP(i)\right)],\label{Eq:Y(i)}\\
\calY'(i)&=\E[\left(\left(I_{M}-\calM\bm\calR_{u}(i+1)\right) \calP(i)\right)\otimes\left(\calM\bmr_{du}(i+1)\right)],\label{Eq:Y'(i)}
\end{align}
and
\begin{align}
&C_1(i)=\E\,\big[\left(I_{M}-\calM\bm\calR_{u}(i+1)\right) f(i)\nn
&\hspace{1.7cm}\times f\T(i)\left(I_{M}-\bm\calR_{u}(i+1)\calM\right)\big],\nn 
&C_2(i)=\E[\left(I_{M}-\calM\bm\calR_{u}(i+1)\right) f(i)\bmr\T_{du}(i+1)]\calM,\nn
&C_3(i)=\E[\calM\bmr_{du}(i+1)f\T(i)\left(I_{M}-\bm\calR_{u}(i+1)\calM\right)],\nn
&\hspace{0.4cm}C_4=\calM\E[\bmr_{du}(i+1)\bmr_{du}\T(i+1)]\calM,\nn
&c(i)=\vect\left(C_1(i)\right)-\vect\left(C_2(i)\right)-\vect\left(C_3(i)\right)+\vect\left(C_4\right),\label{Eq:c(i)}
\end{align}
all of which except $C_4$ depend on the quantities $\{\calP(i),f(i)\}$ at time instant $i$. We also define
\begin{align}
    \calZ=\E[\left(\left(I_{M}-\calM\bm\calR_{u}(i+1)\right) \right)\otimes\left(\left(I_{M}-\calM\bm\calR_{u}(i+1)\right)\right)].\label{Eq:Z(i)}
\end{align}
In the following \cref{Thm:inference privacy preservation}, we evaluate the privacy-preserving performance of the proposed ATP($\delta$) algorithm.

\begin{Theorem}\label{Thm:inference privacy preservation}
Consider the distributed strategy \cref{adaptation,projection,Eq:psi'_l}. Suppose that \cref{Ass:measurement_noise,Ass:Regression_data,Ass:Independence_assumption,Ass:additive_noise} hold. The network inference error at each time instant $i\geq0$ is
\begin{multline}\label{Eq:xi_net_ATP_delta}
    \xi_{\textrm{net}}^{\textrm{ATP($\delta$)}}(i)=\frac{1}{N}\sum_{k=1}^N\frac{1}{|\calN_k|-1}\sum_{\ell\in\calN_k\backslash\{k\}}\Tr\bigg(W_{kk}\\ -U_{k\braces*{k,\ell}}(i)\parens*{X_{{\braces*{k,\ell}}{\braces*{k,\ell}}}(i)+\widecheck{R}_{n,k}(i)}^{-1}U_{k\braces*{k,\ell}}\T(i)\bigg)
\end{multline}
where
\begin{align}
     U_{k\braces*{k,\ell}}(i)&=\brk*{\brk*{\calW-\calV(i)}_{k,k},\brk*{\calW-\calV(i)}_{k,\ell}},\label{Eq:R_wpsi_kkl(i)}\\
\calV(i+1) &= \calV(i)\calB\T(i),\label{Eq:calV}\\
\calV(0) &= \calW\,(I_{M}-\calM\mathcal{R}_{u}),\label{intial}
\end{align}
and where
\begin{dmath}[label={Eq:R_psi_kl(i)}]
          X_{{\braces*{k,\ell}}{\braces*{k,\ell}}}(i)= \left[\!\!\begin{array}{cc}
            \brk*{\Psi(i)\!-\!\E\bpsi(i)\parens*{\E\bpsi(i)}\T}_{k,k}   & \! \brk*{\Psi(i)\!-\!\E\bpsi(i)\parens*{\E\bpsi(i)}\T}_{k,\ell}\\
            \brk*{\Psi(i)\!-\!\E\bpsi(i)\parens*{\E\bpsi(i)}\T}_{\ell,k}   &\! \brk*{\Psi(i)\!-\!\E\bpsi(i)\parens*{\E\bpsi(i)}\T}_{\ell,\ell}
          \end{array}\!\!\right]
\end{dmath}
with
\begin{align}
    \E\bpsi(0)&=\calM\calR_{u}\E\bmw^o,\label{Eq:Epsi_0}\\
    \E\bpsi(i+1) &= \left(I_{M}-\calM\calR_{u}\right) \calP(i)\E\bpsi(i) \nn
    &\hspace{0.4cm}- \left(I_{M}-\calM\calR_{u}\right) f(i)  +\calM\calR_{u}\E\bmw^o,\label{Eq:recursionpsi_i}
\end{align}
and
\begin{align}
    &\Psi(0)=\calM\E[\bm\calR_{u}(0)\,\E[\bmw^o\left(\bmw^o\right)\T]\bm\calR_{u}(0)]\calM+\calM\calG\calM,\label{Eq:vec_Psi_0}\\
    &\vect\left(\Psi(i+1)\right)=\calH(i)\vect\left(\Psi(i)\right)+\calZ\vect(\Gamma(i))+c(i)\nn
    &\hspace{2.55cm}+\left(\calY(i)-\calX(i)-\calX'(i)+\calY'(i)\right)\,\E\bpsi(i).\label{Eq:recursion_vec_Psi_i}
\end{align}
\end{Theorem}
\begin{proof}
See \cref{Sup:evaluation of privacy} of the supplementary material.
\end{proof}

\section{Simulation Results}\label{Sec:Simulation results}
In this section, we compare performance of the proposed ATP$(\delta)$ algorithm in terms of network inference privacy and network MSD against ATP$(0)$, MDA \cite{NasRicFer:J17} and NoCoop \cref{Eq:LMS} in two types of networks: 1) a line network with a large number of linear equality constraints; and 2) a very dense network with a small number of linear equality constraints. We also test the tracking performance of the proposed distributed and adaptive ATP$(\delta)$ algorithm in the case of changing statistics.

\subsection{Selection of Weights for Weighted Projection}
Note that the projection step \cref{projection} at each agent $k$ involves noisy intermediate estimates $\{\bpsi'_{\ell}(i)\}$ from its neighbors $\{\ell\in\calN_k\backslash\{k\}\}$ and the uncontaminated intermediate estimate $\bpsi_k(i)$. In view of the fact that the noisy intermediate estimate $\bpsi'_{\ell}(i)$ involves a privacy noise $\bmn_{\ell}(i)$ with variance $\sigma_{n,\ell}^2(i)$, we set the weights $\braces*{\omega_{\ell k}(i)}$ defined by \cref{Eq:weight_matrix} to
\begin{align}\label{Eq:omega_lk_i}
    \omega_{\ell k}(i)=\left\{\begin{array}{ll}
    \frac{e^{-\sigma_{n,\ell}^2(i)}}{\sum_{\ell\in\calN_k\backslash\{k\}}e^{-\sigma_{n,\ell}^2(i)}+1},     & \text{if $\ell\in\calN_k\backslash\{k\}$}  \\
    \frac{1}{\sum_{\ell\in\calN_k\backslash\{k\}}e^{-\sigma_{n,\ell}^2(i)}+1},     & \text{if $\ell=k$}
    \end{array}
    \right.
\end{align}
which satisfies the requirements in \cref{Eq:weights}.
\subsection{Line Network}
As shown in \cref{Fig:line_network}, we consider a line network which consists of $N=12$ agents. The agents in the network are involved in $Q=11$ linear equality constraints. Let $\calI_q$ be the set of agents involved in the $q$-th constraint for any $q=1,\ldots,Q$. Then in the case of a line network, each $\calI_q$ only contains two neighboring agents that are connected by an edge. Each linear equality constraint is of the form \cite{NasRicFer:J17}:
\begin{align}\label{Eq:linear_eq}
\sum_{k\in\calI_q}d_{qk}\bmw_k^o+b_q\bone_{3\times1}=0_{3\times1},\quad q = 1,\ldots,Q
\end{align}
with the scalar parameters $\{d_{qk},b_q\}$ randomly selected from $[-3,-1]\cup[1,3]$. The lengths of the unknown parameter vectors $\{\bmw_k^o\}$ are $\{M_k\equiv 3\}$.
The random data $\{\bmu_k(i),\bmv_k(i)\}$ are independent, normally distributed with zero mean, and white over time and space. Let
\begin{align*}
\mbox{SNR}_k=10\log_{10}\left(\E[(\bmu_k\T(i)\bmw_k^o)^2]/\sigma_{v,k}^2\right)
\end{align*}
be the signal-to-noise ratio (SNR) at agent $k$. Then, the parameters $\{R_{u,k}, W_{kk}, \E\bmw_k^o,\sigma_{v,k}^2\}$ are adjusted to make $\{\mbox{SNR}_k\}$ as shown in \cref{Fig:SNR_line}. For the step-size parameters, we set $\{\mu_k\equiv0.02\}$ for all the tested algorithms except MDA, where we set $\{\mu_k/j_k\equiv0.02\}$ to ensure the same step-size in the adaptation step for all the above-mentioned algorithms. Let $\delta_k \triangleq \rho\,\trace{W_{kk}}$ with $0\leq\rho\leq1$. Then, for the ATP$(\delta)$ algorithm, we set $\rho=0.1,\,\rho = 0.6,\,\rho=0.85$, respectively, to compare performance under different privacy thresholds $\braces*{\delta_k}$.

\cref{Fig:MSD_line} shows the network MSD learning curves of the tested strategies, which are averaged over 1,000 independent realizations of $\{\bmw_k^o\}$. We observe that 1) simulation results match well with theoretical findings, 2) by adding privacy noises, the proposed ATP$(\delta)$ algorithm converges to a higher MSD level than ATP(0) and MDA \cite{NasRicFer:J17}, 3) larger privacy thresholds result in a higher steady-state network MSD. \cref{Fig:privacy_line} plots the network inference error, defined by \cref{Eq:inference_privacy}. We observe that 1) the network inference privacy performance of ATP(0) and MDA \cite{NasRicFer:J17} are similar, but worse (lower values of $\{\xi_{\textrm{net}}(i)\}$) than the proposed ATP$(\delta)$, 2) by increasing privacy thresholds, the proposed ATP$(\delta)$ algorithm shows better network inference privacy performance. These results demonstrate that the proposed ATP$(\delta)$ algorithm is able to balance the trade-off between network MSD and privacy protection. 
\begin{figure*}[!htb]
    \centering
    \begin{subfigure}[!htb]{0.45\textwidth}
        \centering
        \includegraphics[height=1.7in]{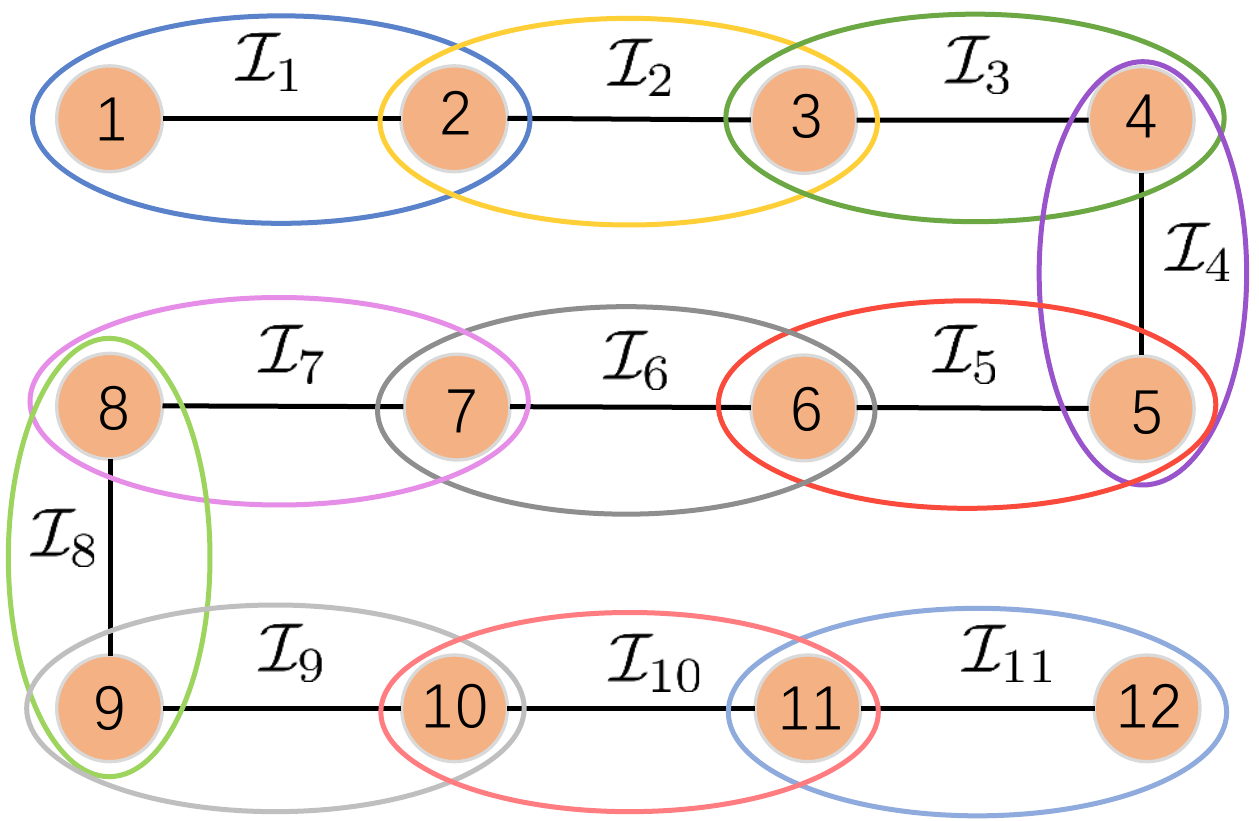}
        \caption{}
        \label{Fig:line_network}
    \end{subfigure}%
    ~ 
    \begin{subfigure}[!htb]{0.5\textwidth}
        \centering
        \includegraphics[height=1.9in]{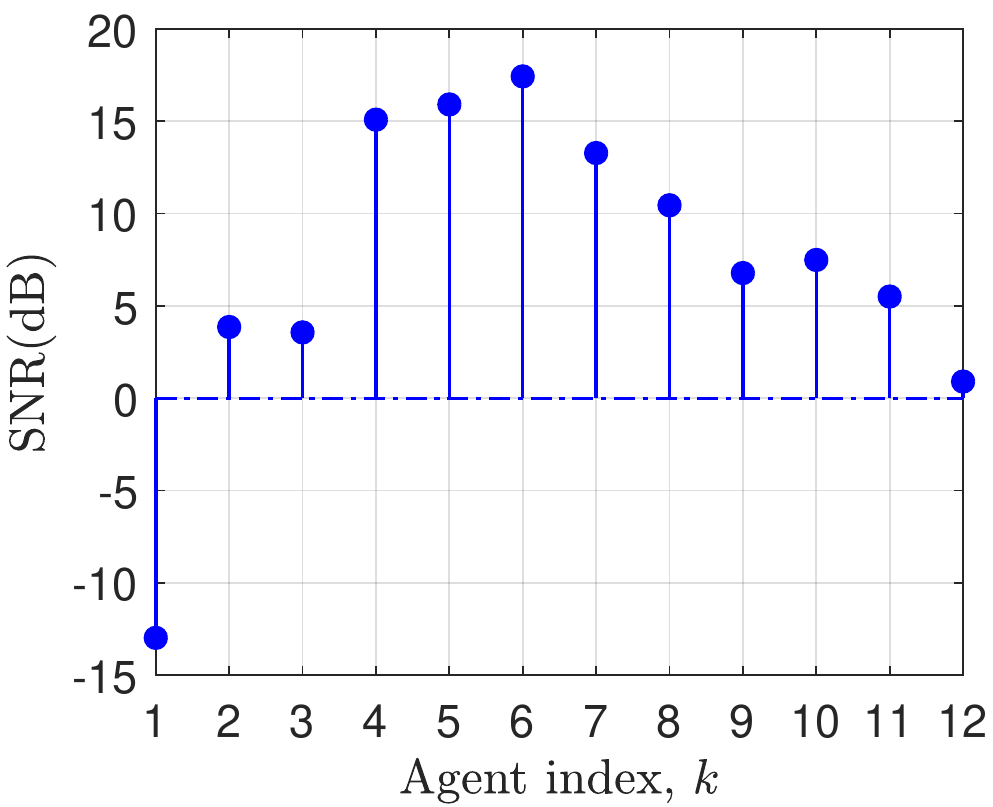}
        \caption{}
        \label{Fig:SNR_line}
    \end{subfigure}
    \caption{(a) A line network and (b) SNRs across the agents in a line network.}
\end{figure*}
\begin{figure}[!htb]
\centering
\includegraphics[width=3.0in]{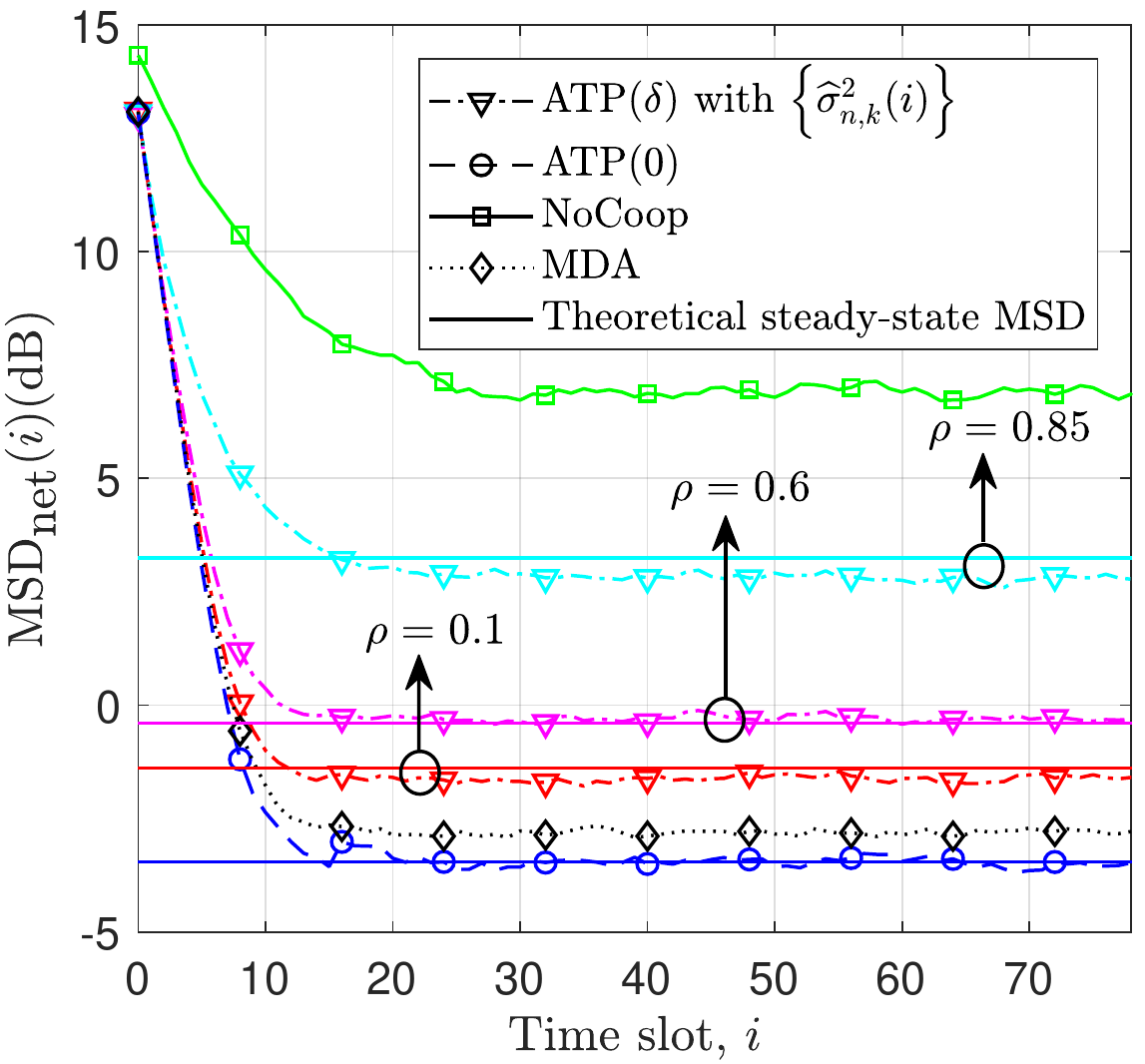}
\caption{Network MSD learning curves of MDA \cite{NasRicFer:J17}, NoCoop, the proposed ATP$(\delta)$ under different parameter settings for $\braces*{\delta_k}$, and ATP(0) in a line network.}
\label{Fig:MSD_line}
\end{figure}
\begin{figure}[!htb]
\centering
\includegraphics[width=3.0in]{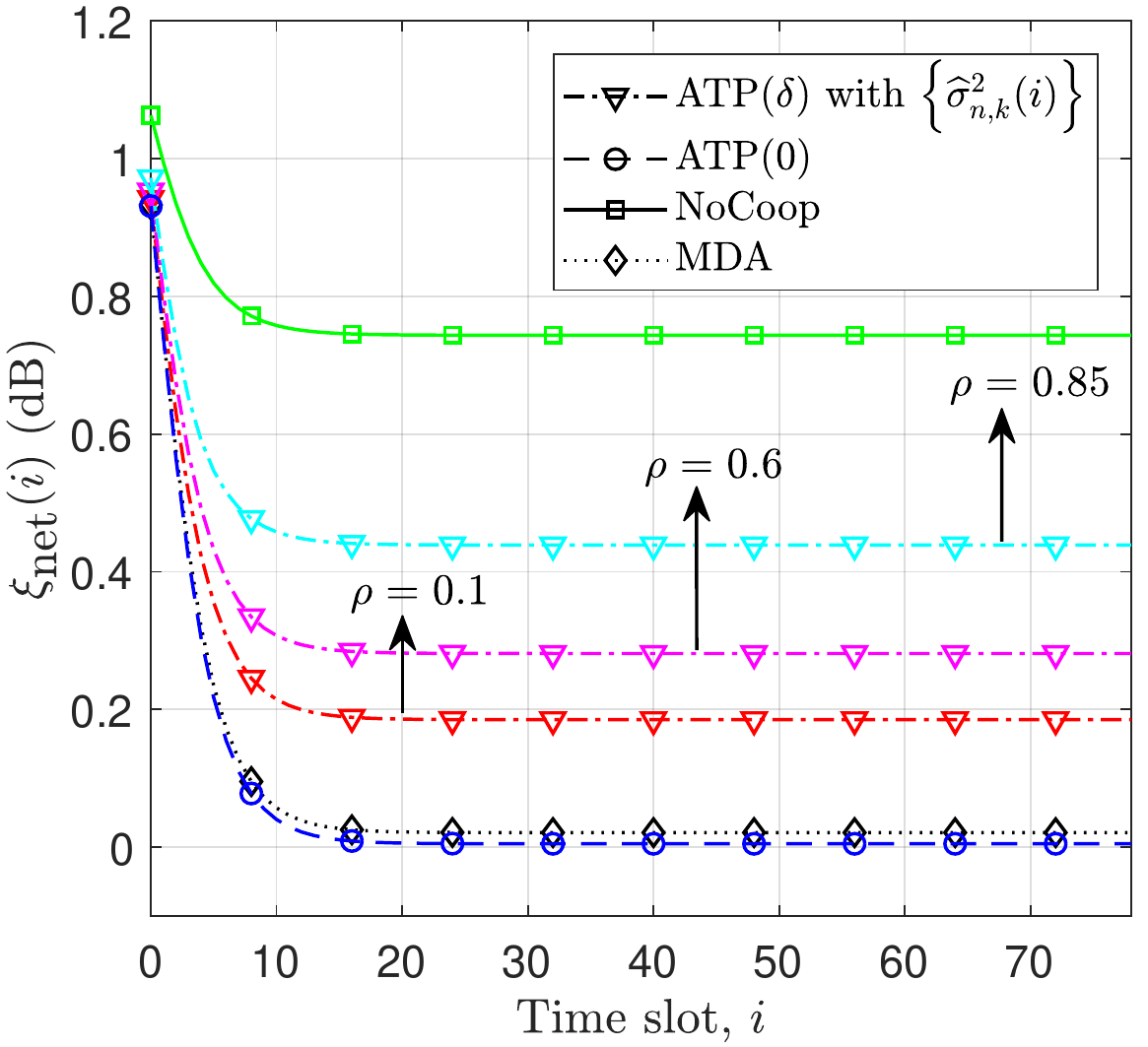}
\caption{Network inference privacy learning curves of MDA \cite{NasRicFer:J17}, NoCoop, the proposed ATP$(\delta)$ under different parameter settings for $\braces*{\delta_k}$, and ATP(0) in a line network. }
\label{Fig:privacy_line}
\end{figure}

\subsection{Dense Network}
As shown in \cref{Fig:dense_network}, we consider a dense network which consists of $N=12$ agents. The agents in the network are involved in $Q=4$ linear equality constraints, each of the form \cref{Eq:linear_eq}. The lengths of the unknown parameter vectors $\{\bmw_k^o\}$ are $\{M_k\equiv 3\}$.
The random data $\{\bmu_k(i),\bmv_k(i)\}$ are independent, normally distributed with zero mean, and white over time and space. The parameters $\{R_{u,k}, W_{kk}, \E\bmw_k^o,\sigma_{v,k}^2\}$ are adjusted to make $\{\mbox{SNR}_k\}$ as shown in \cref{Fig:SNR_dense}. For the ATP$(\delta)$ algorithm, we set the threshold values $\{\delta_k=0.1\,\trace{W_{kk}}\}_{k=1}^N$. The other parameter settings remain the same as those in the previous simulation. 

\cref{Fig:MSD_dense} shows the network MSD learning curves of the tested strategies, which are averaged over 20,000 independent realizations of $\{\bmw_k^o\}$. \cref{Fig:privacy_dense} plots the network inference error, defined by \cref{Eq:inference_privacy}. These results demonstrate that the proposed ATP$(\delta)$ algorithm is still able to balance the trade-off between network MSD and privacy protection in dense networks. \cref{Table:gain_to_loss} lists the ratios of privacy gain to accuracy loss in the line and dense network, where the privacy gain is the absolute gap between the steady-state network inference privacy errors of ATP(0) and ATP$(\delta)$ in dB, and the accuracy loss is the absolute gap between the steady-state network MSDs in dB. We observe that the ratio in the dense network is much higher than in the line network, which means that more privacy gains are obtained in the former case with the same quantity of network MSD loss. These results demonstrate that the proposed ATP$(\delta)$ algorithm performs better in dense networks than in line networks.
\begin{figure*}[!htb]
    \centering
    \begin{subfigure}[!htb]{0.45\textwidth}
        \centering
        \includegraphics[height=2.4in]{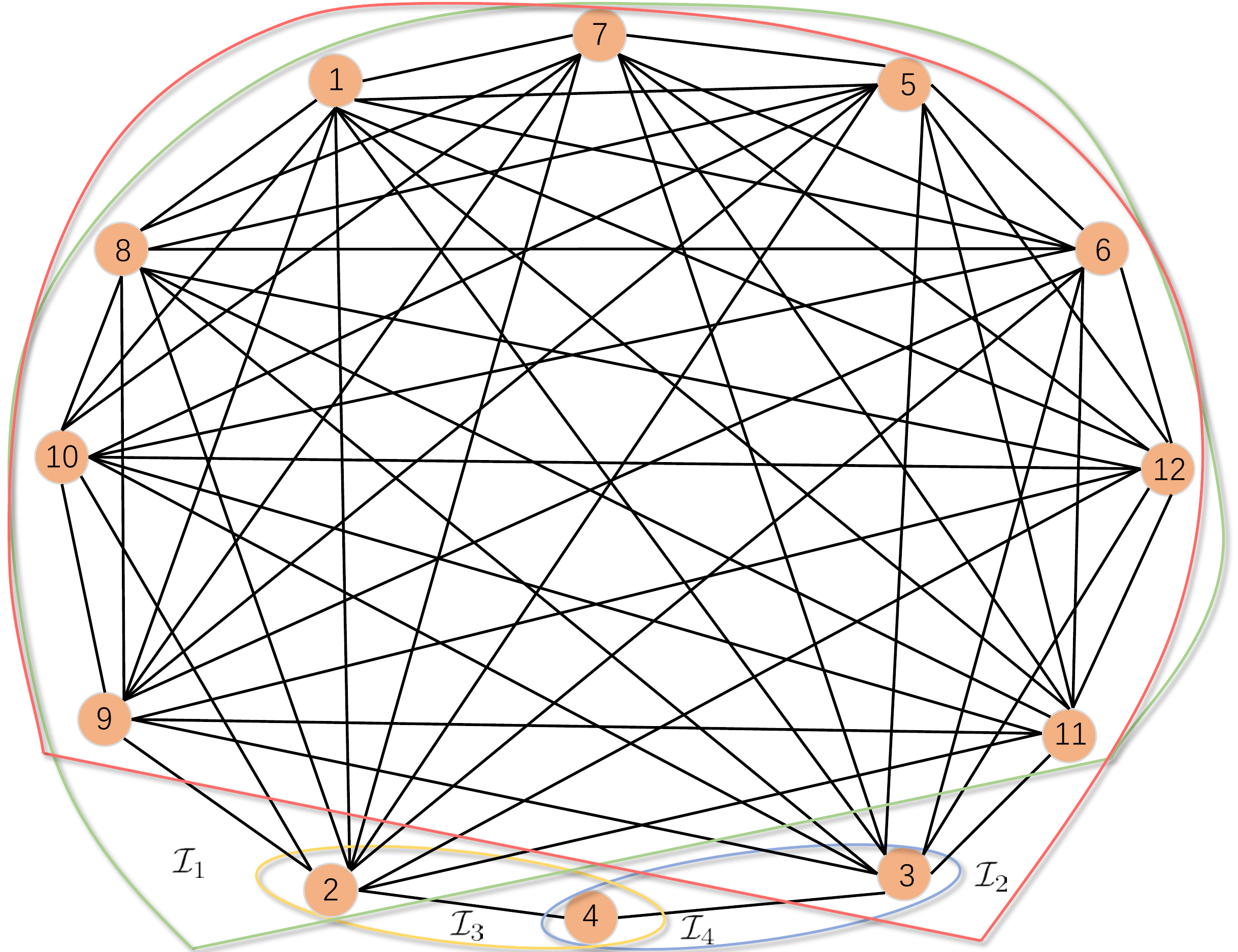}
        \caption{}
        \label{Fig:dense_network}
    \end{subfigure}%
    ~ 
    \begin{subfigure}[!htb]{0.5\textwidth}
        \centering
        \includegraphics[height=2.4in]{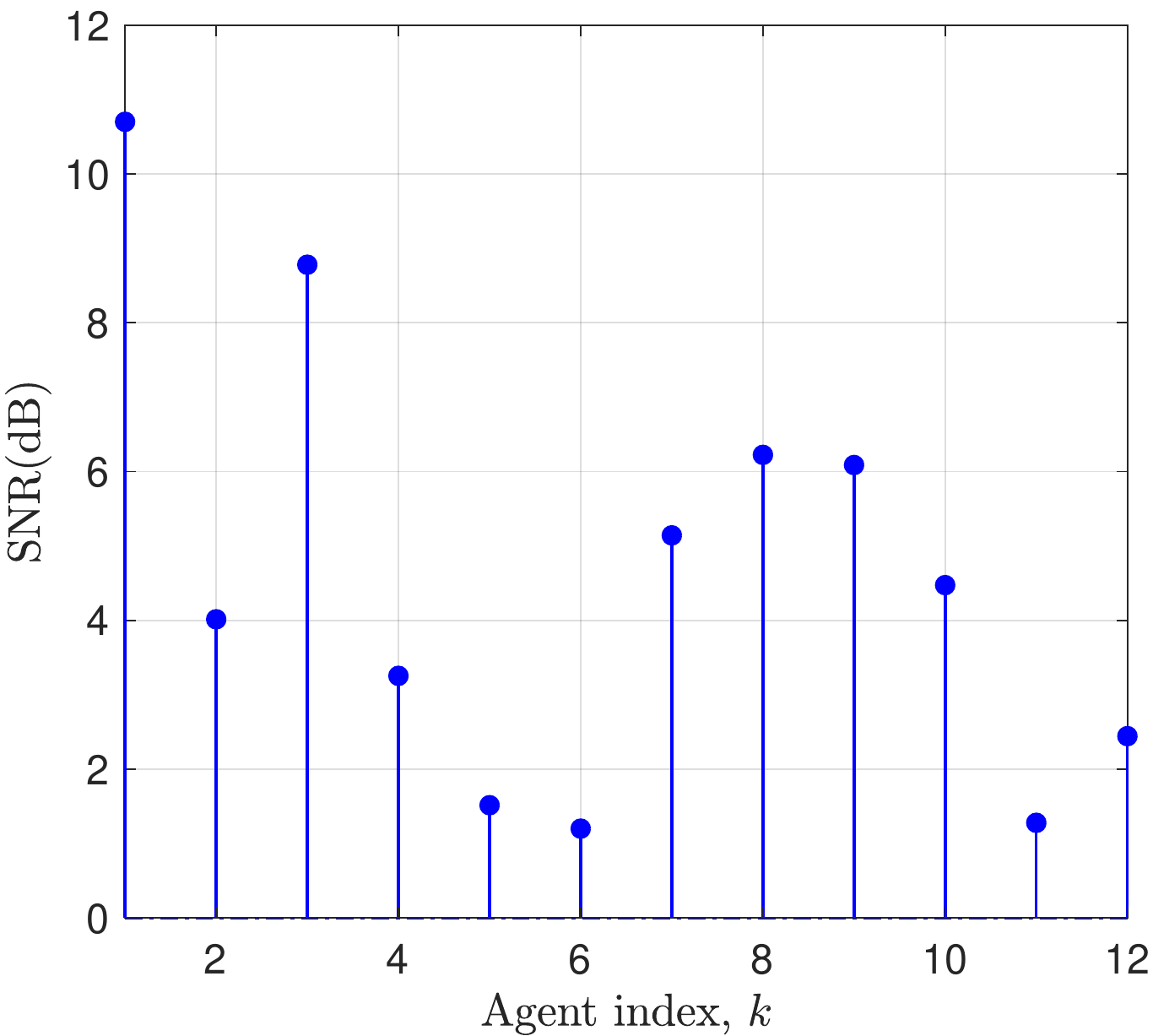}
        \caption{}
        \label{Fig:SNR_dense}
    \end{subfigure}
    \caption{(a) A dense network and (b) SNRs across the agents in a dense network.}
\end{figure*}
\begin{figure}[!htb]
\centering
\includegraphics[width=2.9in]{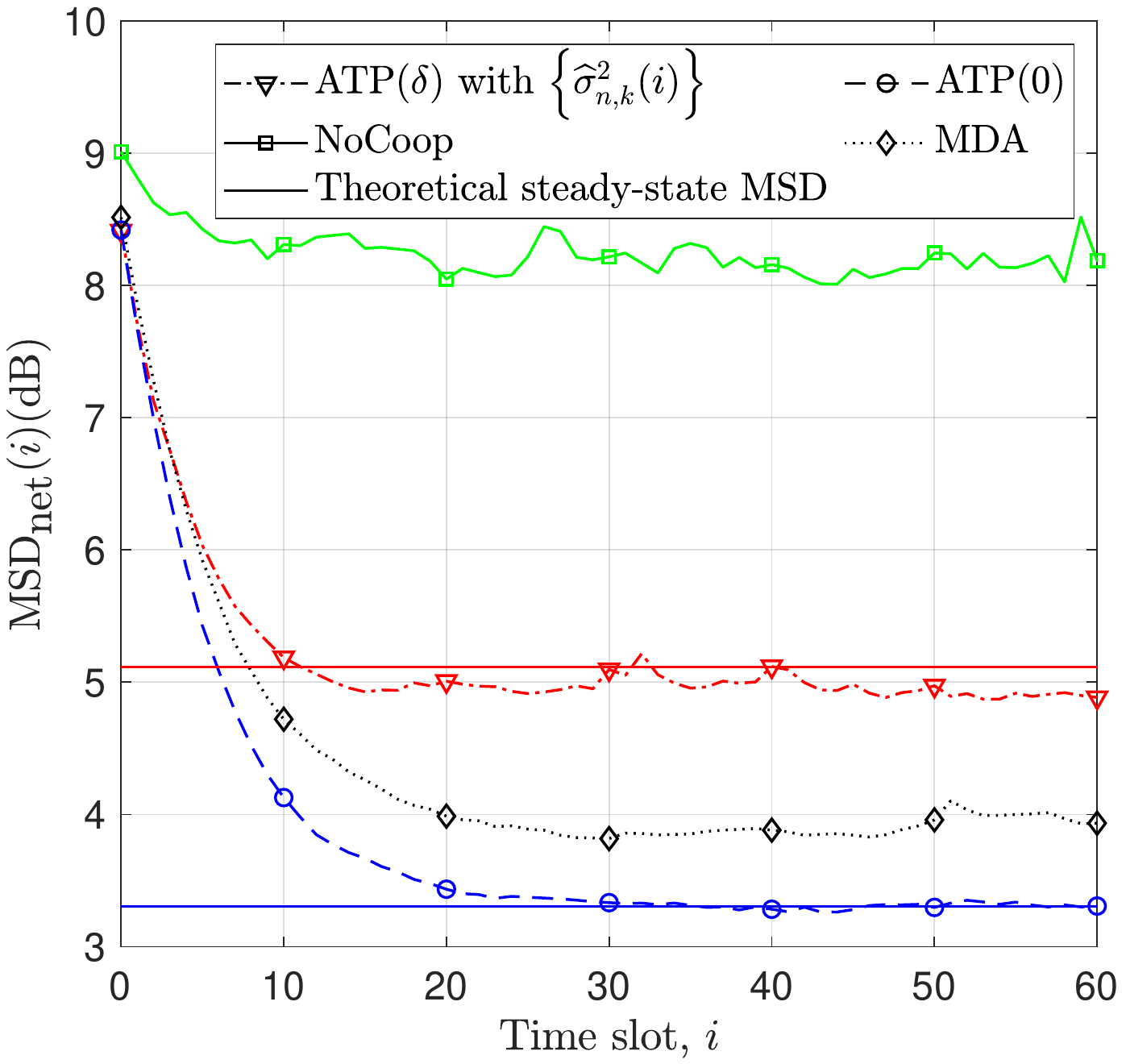}
\caption{Network MSD learning curves of MDA \cite{NasRicFer:J17}, NoCoop, the proposed ATP$(\delta)$ and ATP(0) in a dense network.}
\label{Fig:MSD_dense}
\end{figure}
\begin{figure}[!htb]
\centering
\includegraphics[width=3.0in]{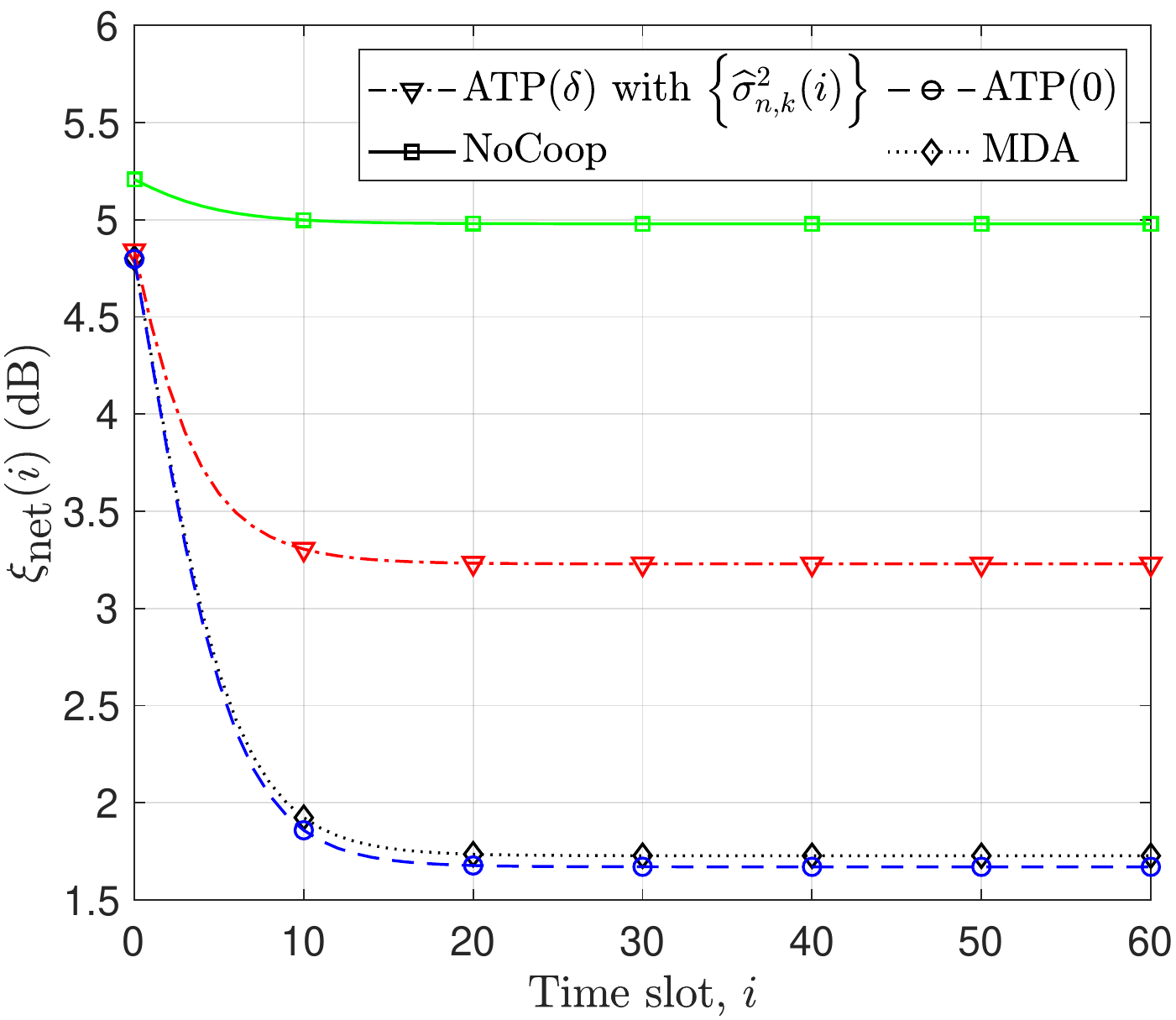}
\caption{Network inference privacy learning curves of MDA \cite{NasRicFer:J17}, NoCoop, the proposed ATP$(\delta)$ and ATP(0) in a dense network. }
\label{Fig:privacy_dense}
\end{figure}
\begin{table}[!htb]
    \centering
    \caption{Ratio of Privacy Gain to Accuracy Loss}
    \label{Table:gain_to_loss}
    \renewcommand{\arraystretch}{1.5}
    \scalebox{1.2}{
    \begin{tabular}{|c|c|c|c|}
    \hline
    \multicolumn{3}{|c|}{Line Network}        &  Dense Network\\
    \cline{1-3}
     $\rho=0.1$    & $\rho=0.6$ & $\rho=0.85$&\\
     \hline
     0.09          &0.09        &  0.06      &0.86\\
     \hline
    \end{tabular}}
\end{table}

\subsection{Tracking Performance}
As shown in \cref{Fig:network_topology}, we consider the case where there are $N=6$ agents in the network. The agents in the network are involved in $Q=5$ linear equality constraints, each of the form \cite{NasRicFer:J17}:
\begin{align*}
\sum_{k\in\calI_q}d_{qk}\bmw_k^o+b_q\bone_{2\times1}=0_{2\times1},\quad q= 1,\ldots,Q
\end{align*}
with the scalar parameters $\{d_{qk},b_q\}$ randomly selected from $[-3,-1]\cup[1,3]$. The lengths of the unknown parameter vectors $\{\bmw_k^o\}$ are $\{M_k\equiv 2\}$.
The random data $\{\bmu_k(i),\bmv_k(i)\}$ are independent, normally distributed with zero mean, and white over time and space. The parameters $\{R_{u,k}, W_{kk}, \E\bmw_k^o,\sigma_{v,k}^2\}$ are adjusted to make $\{\mbox{SNR}_k\}$ as shown in \cref{Fig:SNR_k} at the first stage. Then in order to test the tracking performance of the proposed distributed and adaptive ATP($\delta$) algorithm, we increase values of the elements in $W_{kk}$ at time instant $i=75$, which makes $\{\mbox{SNR}_k\}$ as shown in \cref{Fig:SNR_k_2} at the second stage. For the step-size parameters, we set $\{\mu_k\equiv0.02\}$ for all the tested algorithms except MDA, where we set $\{\mu_k/j_k\equiv0.02\}$ to ensure the same step-size in the adaptation step for all the above-mentioned algorithms. In addition, for the ATP$(\delta)$ algorithm, we set the threshold values $\{\delta_k=0.6\,\trace{W_{kk}}\}_{k=1}^N$ according to the covariance matrix $W_{kk}$ at the first stage. 

\cref{Fig:MSD_net_i} shows the network MSD learning curves of the tested strategies, which are averaged over 1,000 independent runs. We observe that by increasing values of the elements in $W_{kk}$, the proposed ATP$(\delta)$ algorithm with privacy noise powers $\braces*{\sigma_{n,k}^2(i)=\widehat{\sigma}_{n,k}^2(\infty)}$ converges to a higher MSD level at the second stage. This is because privacy noise powers become larger at the second stage according to \cref{Eq:lim_sigma_n_k_i}. We also note that the steady-state network MSD of the proposed distributed and adaptive ATP$(\delta)$ gets higher at the second stage. \cref{Fig:privacy} plots the network inference error, defined by \cref{Eq:inference_privacy}. We observe that the network inference privacy performance of both ATP$(\delta)$ algorithms are similar throughout the whole process. These results demonstrate that the proposed distributed and adaptive ATP$(\delta)$ algorithm is able to track changes of the covariance matrix $W_{kk}$. 
\begin{figure*}[!htb]
    \centering
    \begin{subfigure}[!htb]{0.3\textwidth}
        \centering
        \includegraphics[height=1.5in]{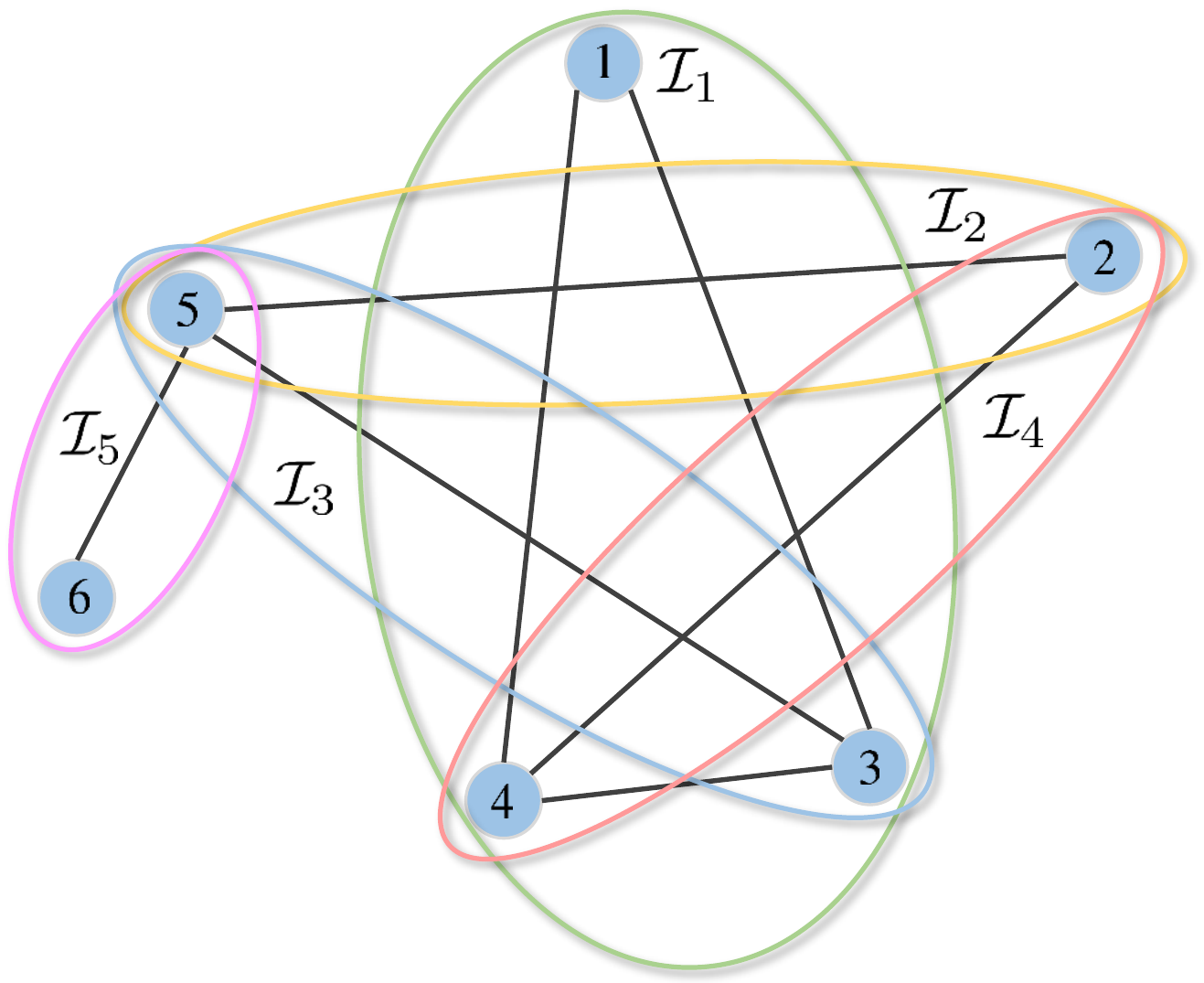}
        \caption{}
        \label{Fig:network_topology}
    \end{subfigure}%
    ~ 
    \begin{subfigure}[!htb]{0.3\textwidth}
        \centering
        \includegraphics[height=1.5in]{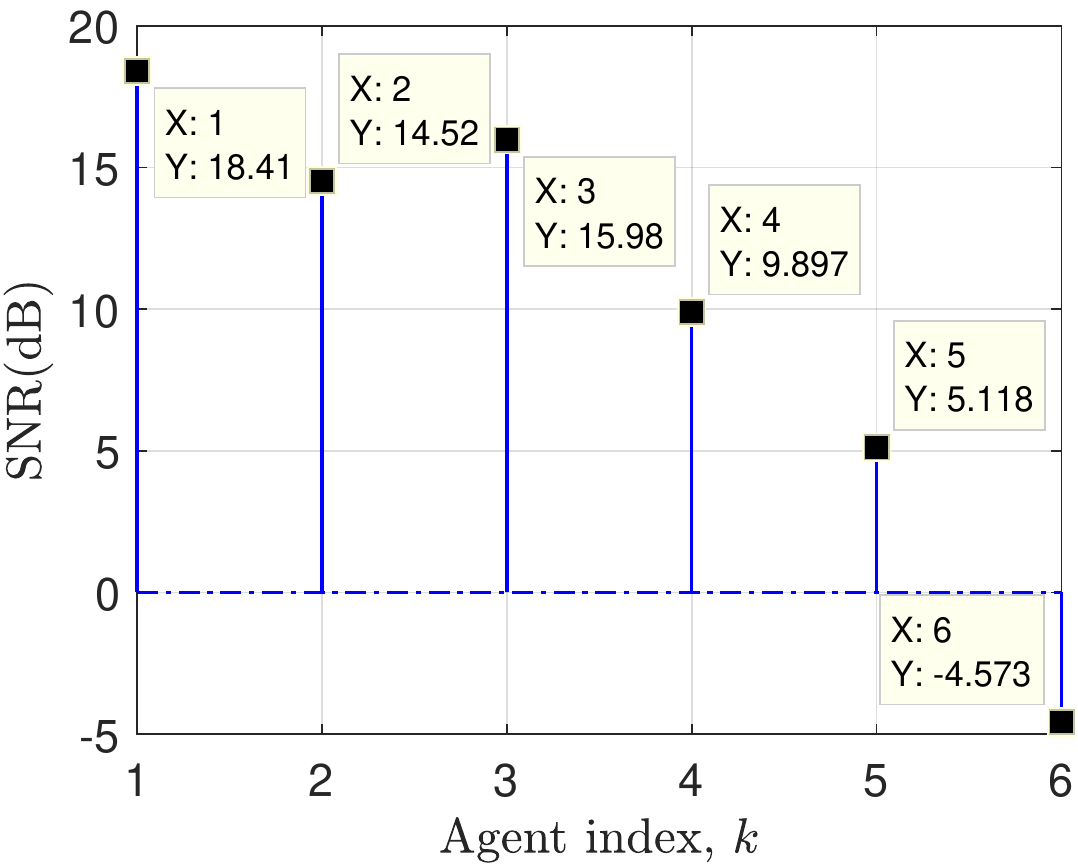}
        \caption{}
        \label{Fig:SNR_k}
    \end{subfigure}
    ~
    \begin{subfigure}[!htb]{0.3\textwidth}
        \centering
        \includegraphics[height=1.5in]{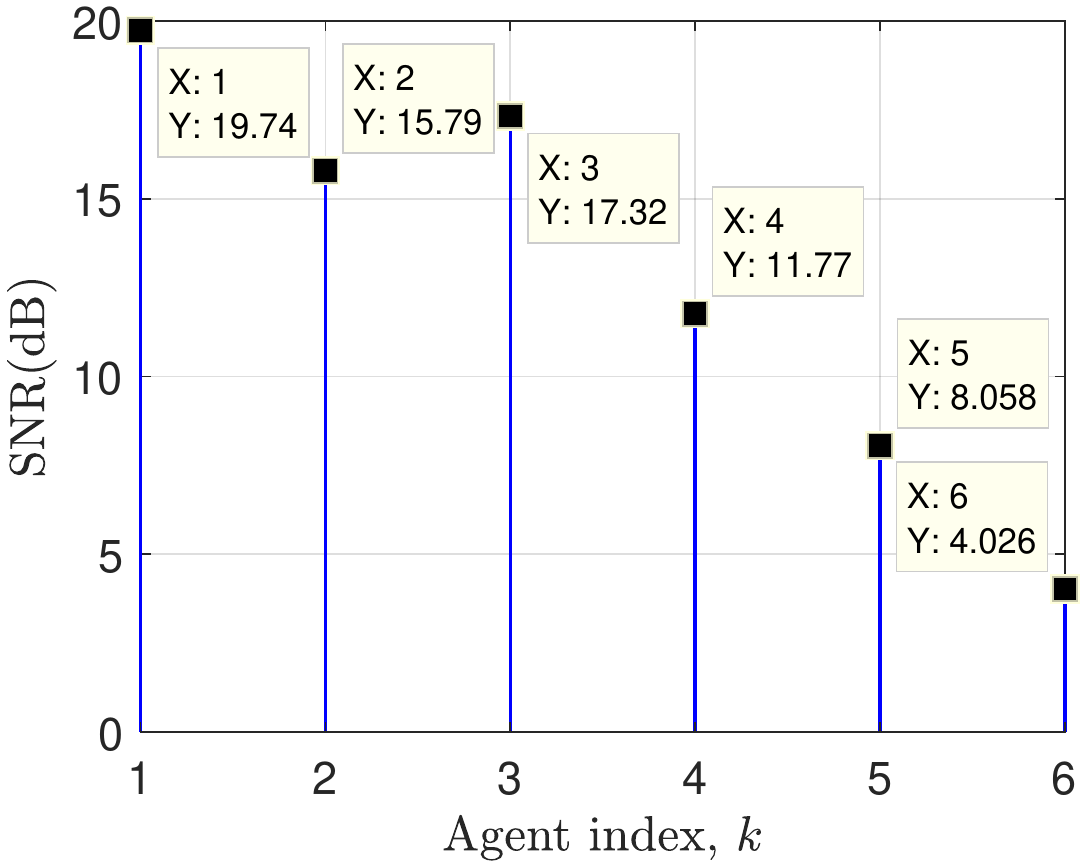}
        \caption{}
        \label{Fig:SNR_k_2}
    \end{subfigure}
    \caption{(a) Network topology, (b) SNRs across the agents at the first stage and (c) SNRs across the agents at the second stage.}
\end{figure*}
\begin{figure}[!htb]
\centering
\includegraphics[width=3.0in]{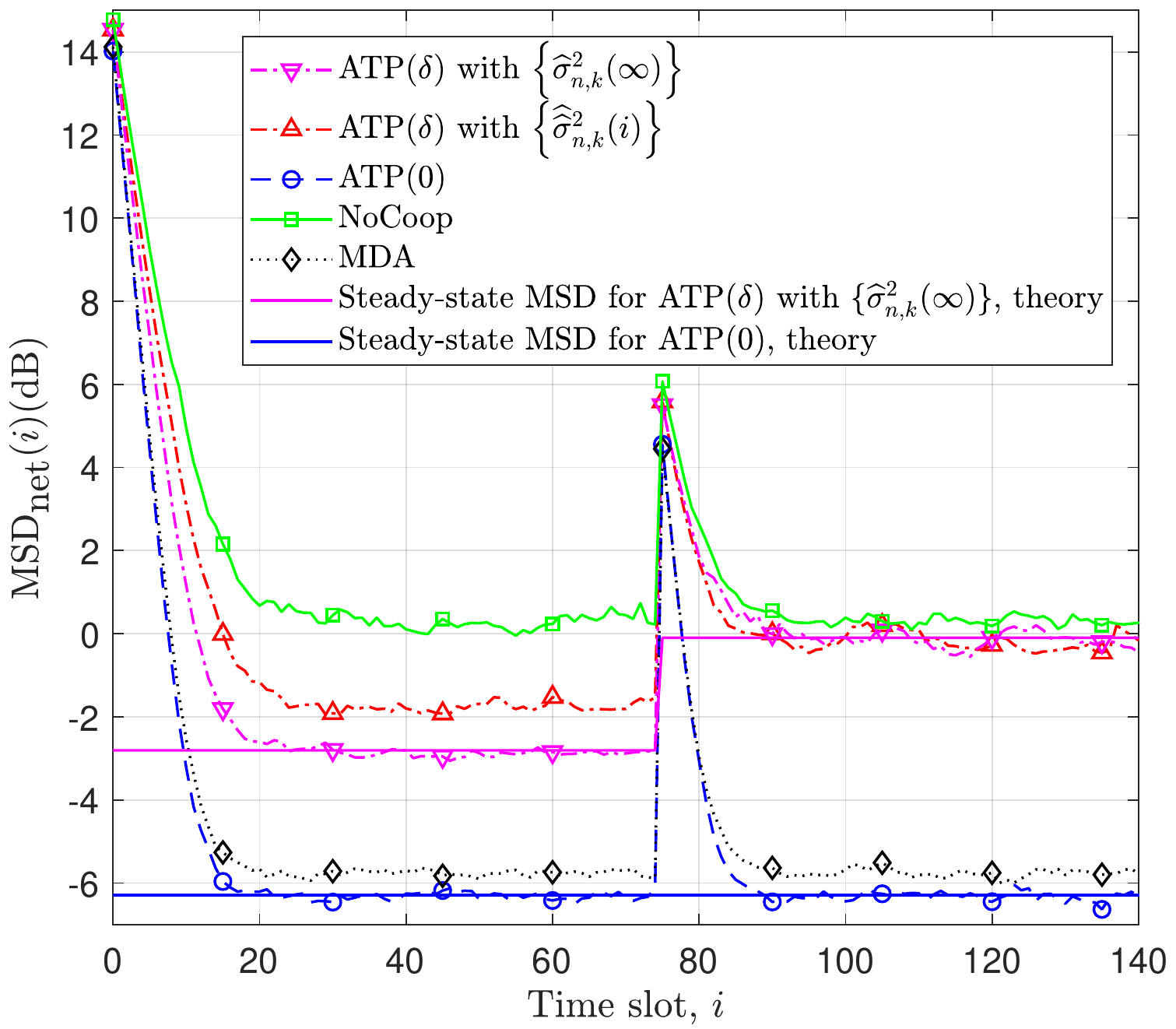}
\caption{Network MSD learning curves of MDA \cite{NasRicFer:J17}, NoCoop, the proposed ATP$(\delta)$ with privacy noise powers $\braces*{\hsigma_{n,k}^2(\infty)}$ as shown by \cref{Eq:lim_sigma_n_k_i} and $\braces*{\widehat{\hsigma}_{n,k}^2(i)}$ as shown by \cref{Eq:adaptive_privacy_noise_power}, and ATP(0).}
\label{Fig:MSD_net_i}
\end{figure}
\begin{figure}[!htb]
\centering
\includegraphics[width=3.0in]{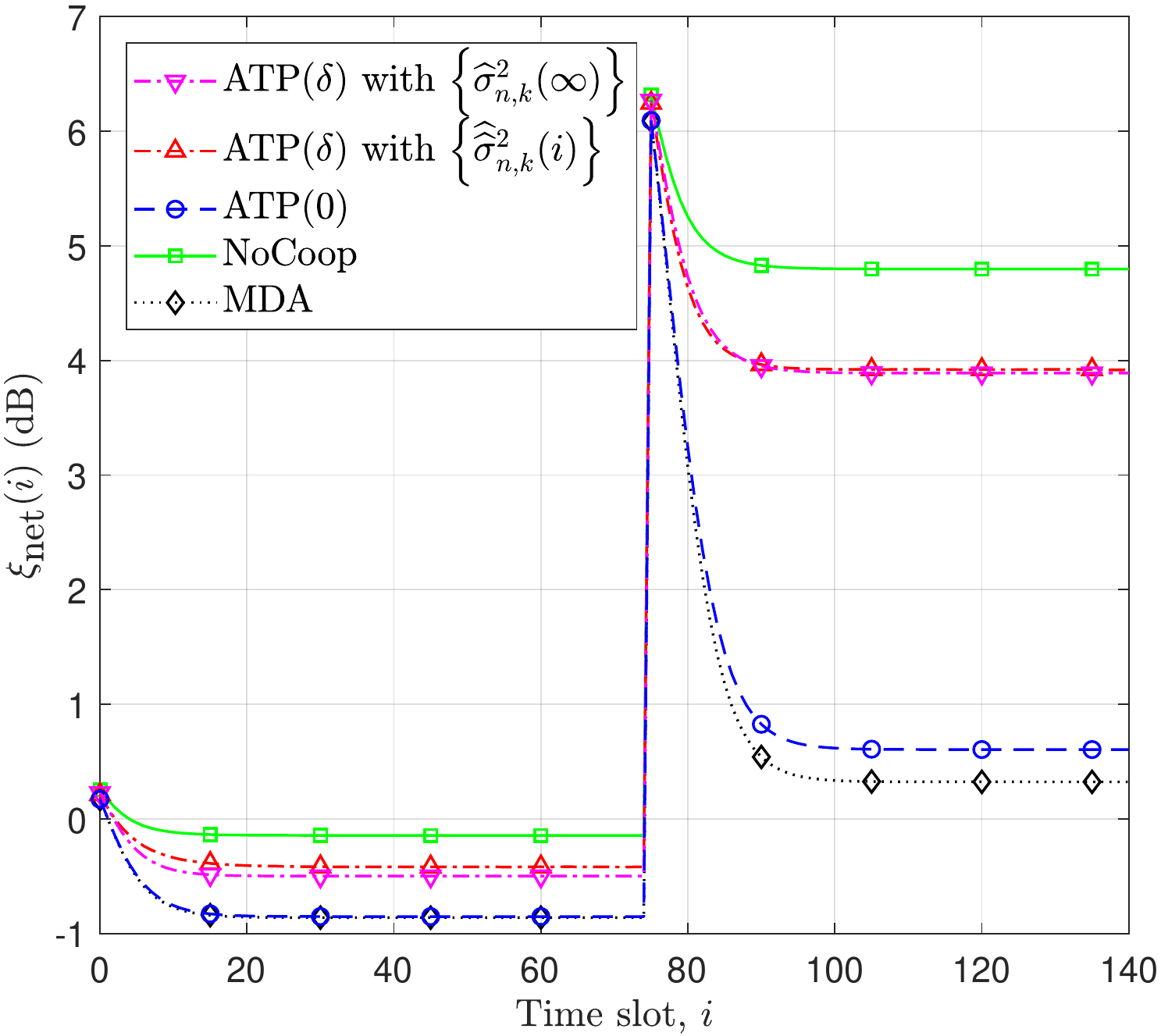}
\caption{Network inference privacy learning curves of MDA \cite{NasRicFer:J17}, NoCoop, the proposed ATP$(\delta)$ with privacy noise powers $\braces*{\hsigma_{n,k}^2(\infty)}$ as shown by \cref{Eq:lim_sigma_n_k_i} and $\braces*{\widehat{\hsigma}_{n,k}^2(i)}$ as shown by \cref{Eq:adaptive_privacy_noise_power}, and ATP(0). }
\label{Fig:privacy}
\end{figure}

\section{Conclusion}\label{Sect: Conclusion}
We have developed a privacy-preserving distributed strategy over linear multitask networks, which is able to protect each agent's local task by adding a privacy noise to its local information before sharing with its neighbors. We proposed a utility-privacy optimization trade-off to determine the amount of noise to add, and derived a sufficient condition for the privacy noise powers in order to satisfy the proposed privacy constraints. Furthermore, we proposed a distributed and adaptive scheme to compute the privacy noise powers. We have studied the mean and mean-square behaviors and privacy-preserving performance. Simulation results demonstrated that the proposed scheme is able to balance the trade-off between the network MSD and network inference privacy. 

Future work concerns privacy issues over multitask networks where agents' local parameters of interest consist of common parameters within each neighborhood and individual parameters. Privacy-preserving distributed strategies will be developed to improve the estimation accuracy of common parameters at each agent via in-network cooperation with neighboring agents, as well as to protect individual parameters against privacy leakage.

\bibliographystyle{IEEEtran}
\bibliography{IEEEabrv,StringDefinitions,refs}

\newpage
\setcounter{page}{1}
\renewcommand{\thesection}{\arabic{section}}
\setcounter{section}{0}
\section{Proof of \cref{Thm:Sufficient condition}}\label{Sup:sufficient condition}

The $M_k\times M_k$ symmetric positive semi-definite matrix $X_{kk}(i)$ admits a spectral decomposition of the form:
\begin{align*}
X_{kk}(i) = T_k(i)\Lambda_k(i)T\T_k(i),
\end{align*}
where $T_k(i)$ is an orthogonal matrix and $\Lambda_k(i) = \diag\left\{\lambda_k^m(i)\right\}_{m=1}^{M_k}$ is a diagonal matrix consisting of the eigenvalues of $X_{kk}(i)$. Let 
\begin{align}\label{tildeU}
\widetilde{U}_{kk}(i)=T\T_k(i)U_{kk}\T(i) U_{kk}(i)T_k(i),
\end{align}
which is a symmetric positive semi-definite matrix. Let $\{\widetilde{u}_{kk}^m(i)\}_{m=1}^{M_k}$ be entries on the main diagonal of $\widetilde{U}_{kk}(i)$. Then, it follows that $\widetilde{u}_{kk}^m(i)\geq0$ for all $m =1,\ldots, M_k$. From the \gls{LHS} of (\ref{rewrittenCost}), we have 
\begin{align}
& \trace{U_{kk}(i)\left(X_{kk}(i) + \sigma_{n,k}^2(i)I_{M_k}\right)^{-1}U_{kk}\T(i)} \nn
& =\trace{U_{kk}(i)T_k(i)\parens*{\Lambda_k(i) + \sigma_{n,k}^2(i)I_{M_k}}^{-1}T\T_k(i)U_{kk}\T(i)}\nonumber\\
&\stackrel{\scriptstyle{\textrm{(a)}}}=\trace{\parens*{\Lambda_k(i) + \sigma_{n,k}^2(i)I_{M_k}}^{-1}\widetilde{U}_{kk}(i)}\nonumber\\
&=\sum_{m=1}^{M_k}\frac{\widetilde{u}_{kk}^m(i)}{\lambda_k^m(i)+\sigma_{n,k}^2(i)}\nn
&\leq \sum_{m=1}^{M_k}\frac{\widetilde{u}_{kk}^m(i)}{\sigma_{n,k}^2(i)}\nonumber\\
&=\frac{\trace{\widetilde{U}_{kk}(i)}}{\sigma_{n,k}^2(i)}\nonumber\\
&\stackrel{\fr{tildeU}}=\frac{ \trace{ T\T_k(i)U_{kk}\T(i) U_{kk}(i)T_k(i) } }{\sigma_{n,k}^2(i)}\nonumber\\
&=\frac{ \trace{ U_{kk}\T(i) U_{kk}(i) } }{\sigma_{n,k}^2(i)},\nonumber
\end{align}
where in step (a) we used the property $\trace{AB}=\trace{BA}$ for any square matrices $A$ and $B$ of compatible sizes. The theorem now follows from \cref{rewrittenCost}.

\section{Proof of \cref{Lem:steady-state variance}}\label{Sup:proof of steady-state variance}

We first show a preliminary result.
\begin{Lemma_A}\label{lem:calV}
Consider the matrix $\calV(i)$ defined by \cref{Eq:calV(i)}. If the step-size $\mu_k$ satisfies \cref{Cond:variance convergence} for all agents $k=1,\cdots,N$. Then, $\lim_{i\to\infty} \calV(i) = 0_{M\times M}$.
\end{Lemma_A}
\begin{IEEEproof}
We define $\calB(i)$ as shown by \cref{Eq:calB(i)}. Under \cref{Ass:measurement_noise,Ass:Independence_assumption,Ass:additive_noise}, it follows from (\ref{Eq:psierror_dynamics}) that 
\begin{align*}
\E[\bmw^o\tbpsi\T(i+1)] &=\E[\bmw^o\tbpsi\T(i)]\calB\T(i)\nn
\E\tbpsi(i+1) &= \calB(i)\E\tbpsi(i)
\end{align*}
for any $i\geq 0$, and
\begin{align}
\calV(i+1)
&=\E[(\bmw^o-\E\bmw^o)(\tbpsi(i+1)-\E\tbpsi(i+1))\T]\nn
&= \E[\bmw^o\tbpsi\T(i+1)]-\E\bmw^o\brk*{\E\tbpsi(i+1)}\T \label{Eq:CalV_split}\\
&= \E[\bmw^o\tbpsi\T(i)]\calB\T(i) - \E\bmw^o\brk*{\E\tbpsi(i)}\T\calB\T(i)\nn
&= \calV(i)\calB\T(i).\label{eq:calV}
\end{align}
We next show that $\calB(i)$ is stable, i.e., $\rho\parens*{\calB(i)}<1$. We have 
\begin{align*}
\rho\parens*{\calB(i)}
&\leq\norm{\calB(i)}\nn
&\leq\norm{\calP(i)}\,\norm{I_{M}-\calM\calR_{u}}.
\end{align*}
From \cref{Eq:calRue,Eq:calM}, since $I_{M}-\calM\calR_{u}$ is a symmetric matrix, we have
\begin{align*}
\norm{I_{M}-\calM\calR_{u}}=\rho\parens*{I_{M}-\calM\calR_{u}}.
\end{align*}
The eigenvalues of $I_M-\calM\calR_u$ are $\{1-\mu_k\lambda(R_{u,k})\}_{k=1}^N$, thus $\calB(i)$ is stable if
\begin{align}\label{Ineq:mean stabilitiy1}
\Big|1-\mu_k\lambda(R_{u,k})\Big|<\frac{1}{\norm{\calP(i)}}
\end{align}
holds for all agents $k=1,\ldots,N$. By solving the inequality in \cref{Ineq:mean stabilitiy1} for $\mu_k$, we arrive at the desired sufficient condition \cref{Cond:variance convergence}. From \cref{eq:calV}, the lemma follows.
\end{IEEEproof}
We return to the proof of \cref{Lem:steady-state variance}. We have
\begin{align}\label{Eq:U_kk_calculation}
U_{kk}(i) &\stackrel{\ft{(a)}}=\E[\left(\bmw_k^o-\E\bmw_k^o\right)\left(\bpsi_k(i)-\E\bpsi_k(i)\right)\T]\nn
&\stackrel{\ft{(b)}}
=W_{kk}-\brk*{\calV(i)}_{k,k},
\end{align}
where in (a) we used \cref{Ass:additive_noise} and in (b) the $W_{kk}$ is defined by \cref{Wkk}. From \cref{lem:calV}, as $i\to\infty$, $U_{kk}(i)\to W_{kk}$. The proposition now follows from \cref{Eq:sigma_nk_2}.

\section{Proof of \cref{Thm:inference privacy preservation}}\label{Sup:evaluation of privacy}

Under \cref{Ass:additive_noise}, it follows from \cite[p.66]{Say:B08} that
\begin{multline}\label{Eq:w_hat_coop_noise}
\E\norm{\bmw_k^o-{\hbmw}_{k\mid\{\psi'_k,\psi_\ell\}}(i)}^2=
\Tr\bigg(W_{kk} \\-U_{k\braces*{k,\ell}}(i)\parens*{X_{{\braces*{k,\ell}}{\braces*{k,\ell}}}(i)+\widecheck{R}_{n,k}(i)}^{-1}U_{k\braces*{k,\ell}}\T(i)\bigg)
\end{multline}
which leads to \cref{Eq:xi_net_ATP_delta}. From \cref{Eq:U_k_kl,Eq:calV(i),Eq:calW}, we have \cref{Eq:R_wpsi_kkl(i)}. The recursion \cref{Eq:calV} is shown by \cref{eq:calV} in \cref{Sup:proof of steady-state variance}. Given \cref{Ass:measurement_noise,Ass:Independence_assumption,Ass:Regression_data}, it follows from \cref{Eq:adaptation_net} that
\begin{align}\label{Eq:EPsi_0}
\E\tbpsi(0)=(I_{M}-\calM\calR_{u})\E\bmw^o,
\end{align}
and
\begin{align}\label{Eq:Ewo_tpsi0}
\E[\bmw^o\tbpsi\T(0)] = \E[\bmw^o\left(\bmw^o\right)\T](I-\calM\calR_{u}).
\end{align}
We then arrive at the desired result \cref{intial} from \cref{Eq:CalV_split}. From \cref{Eq:Psi(i),Eq:X_kl_kl}, we have (\ref{Eq:R_psi_kl(i)}), and from \cref{Eq:EPsi_0}, we have \cref{Eq:Epsi_0}. 
It follows from \cref{adaptation,projection,Eq:psi'_l} that
\begin{dmath}[label={Eq:recursion_psi_i}]
\bpsi(i+1) = \left(I_{M}-\calM\bm\calR_{u}(i+1)\right) \calP(i)\bpsi(i) - \left(I_{M}-\calM\bm\calR_{u}(i+1)\right) f(i) + \left(I_{M}-\calM\bm\calR_{u}(i+1)\right) \bmq(i) +\calM\bmr_{du}(i+1)
\end{dmath}
for any time instant $i\geq 0$, with
\begin{align}\label{Eq:psi0}
\bpsi(0) = \calM\bmr_{du}(0),
\end{align}
and where the quantity $f(i)$ is defined by \cref{Eq:f(i)}. Taking expectations on both sides of (\ref{Eq:recursion_psi_i}) gives \cref{Eq:recursionpsi_i}
for any time instant $i\geq 0$, and where we used \cref{Ass:measurement_noise,Ass:Regression_data,Ass:Independence_assumption,Ass:additive_noise,Eq:Eq(i),Eq:Eg(i)}. Multiply both sides of (\ref{Eq:recursion_psi_i}) by its transpose from the right before taking expectations, then we have
\begin{dmath*}
\Psi(i+1) = \E\,\big[\left(I_{M}-\calM\bm\calR_{u}(i+1)\right) \calP(i)\Psi(i)
\calP\T(i)\left(I_{M}-\bm\calR_{u}(i+1)\calM\right)\big] - \E\,\big[\left(I_{M}-\calM\bm\calR_{u}(i+1)\right)\calP(i)\E\bpsi(i)f\T(i)\left(I_{M}-\bm\calR_{u}(i+1)\calM\right)\big] + \E[\left(I_{M}-\calM\bm\calR_{u}(i+1)\right)\calP(i)\E\bpsi(i)\bmr_{du}\T(i+1)]\calM- \E\,\big[\left(I_{M}-\calM\bm\calR_{u}(i+1)\right) f(i)\E\bpsi\T(i)\calP\T(i)\left(I_{M}-\bm\calR_{u}(i+1)\calM\right)\big] + C_1(i) 
- C_2(i)+ \E[\left(I_{M}-\calM\bm\calR_{u}(i+1)\right) \Gamma(i)\left(I_{M}-\bm\calR_{u}(i+1)\calM\right)] 
+\E[\calM\bmr_{du}(i+1)\E\bpsi\T(i)\calP\T(i)\left(I_{M}-\bm\calR_{u}(i+1)\calM\right)]-C_3(i)+C_4
\end{dmath*}
which makes use of \cref{Ass:measurement_noise,Ass:Regression_data,Ass:Independence_assumption,Ass:additive_noise}. We then arrive at the desired result (\ref{Eq:recursion_vec_Psi_i}) for any time instant $i\geq 0$ by using the property in \cref{Eq:vect_product}. Initially at $i=0$, we have
\begin{align*}
\Psi(0) &\stackrel{\scriptstyle{\cref{Eq:psi0}}}=\calM\E[\bmr_{du}(0)\bmr\T_{du}(0)]\calM\nn
&\stackrel{\fr{Eq:r_du}}=\calM\E[\left(\bm\calR_{u}(0)\bmw^o+\bmg(0)\right)\left(\bm\calR_{u}(0)\bmw^o+\bmg(0)\right)\T]\calM\nn
&\stackrel{\ft{(a)}}=\calM\E[\bm\calR_{u}(0)\bmw^o\left(\bmw^o\right)\T\bm\calR_{u}(0)]\calM\nn
&\hspace{0.4cm}+\calM\E[\bmg(0)\left(\bmg(0)\right)\T]\calM\nn
&\stackrel{\ft{(b)}}=\calM\E[\bm\calR_{u}(0)\,\E[\bmw^o\left(\bmw^o\right)\T]\bm\calR_{u}(0)]\calM+\calM\calG\calM,
\end{align*}
which is the desired result \cref{Eq:vec_Psi_0}, and where in (a) we used \cref{Ass:measurement_noise,Ass:Independence_assumption}, and in (b) we used \cref{Ass:Independence_assumption,Eq:calG}. The proof is now complete.

\end{document}

%% file: preamble.tex
\usepackage[T1]{fontenc}
\usepackage{amsmath,amssymb,amsfonts,mathrsfs,bm}
\usepackage{mathtools}
\usepackage{amsthm}
\usepackage{nicefrac}
\usepackage{cite}
\usepackage[shortlabels]{enumitem}
\usepackage{graphicx}
\usepackage{epstopdf}
\usepackage{url}
\usepackage{colortbl}
\usepackage{booktabs}
\usepackage{multirow}
\usepackage[table,dvipsnames]{xcolor}
\usepackage[normalem]{ulem}

\usepackage{array}
\newcolumntype{L}[1]{>{\raggedright\let\newline\\\arraybackslash\hspace{0pt}}m{#1}}
\newcolumntype{C}[1]{>{\centering\let\newline\\\arraybackslash\hspace{0pt}}m{#1}}
\newcolumntype{R}[1]{>{\raggedleft\let\newline\\\arraybackslash\hspace{0pt}}m{#1}}

\makeatletter
\let\MYcaption\@makecaption
\makeatother
\usepackage[font=footnotesize]{subcaption}
\makeatletter
\let\@makecaption\MYcaption
\makeatother

\usepackage{xparse}



\usepackage{glossaries}

\makeatletter
\let\oldgls\gls
\let\oldglspl\glspl

\newcommand\fussy@ifnextchar[3]{%
  \let\reserved@d=#1%
  \def\reserved@a{#2}%
  \def\reserved@b{#3}%
  \futurelet\@let@token\fussy@ifnch}
\def\fussy@ifnch{%
  \ifx\@let@token\reserved@d
    \let\reserved@c\reserved@a 
  \else
    \let\reserved@c\reserved@b
  \fi
 \reserved@c}

\renewcommand{\gls}[1]{%
  \oldgls{#1}\fussy@ifnextchar.{\@checkperiod}{\@}}
\renewcommand{\glspl}[1]{%
  \oldglspl{#1}\fussy@ifnextchar.{\@checkperiod}{\@}}

\newcommand{\@checkperiod}[1]{%
  \ifnum\sfcode`\.=\spacefactor\else#1\fi
}
\makeatother

\newacronym{wrt}{w.r.t.}{with respect to}
\newacronym{RHS}{RHS}{right-hand side}
\newacronym{LHS}{LHS}{left-hand side}
\newacronym{iid}{i.i.d.}{independent and identically distributed}

\usepackage{float}

\ifx\notloadhyperref\undefined
	\ifx\loadbibentry\undefined
		\usepackage[hidelinks,hypertexnames=false]{hyperref} 
	\else
		\usepackage{bibentry}
		\makeatletter\let\saved@bibitem\@bibitem\makeatother
		\usepackage[hidelinks,hypertexnames=false]{hyperref}
		\makeatletter\let\@bibitem\saved@bibitem\makeatother
	\fi
\else
	\relax
\fi

\usepackage[capitalize]{cleveref}
\crefname{equation}{}{}
\Crefname{equation}{}{}
\crefname{claim}{claim}{claims}
\crefname{step}{step}{steps}
\crefname{line}{line}{lines}
\crefname{condition}{condition}{conditions}
\crefname{dmath}{}{}
\crefname{dseries}{}{}
\crefname{dgroup}{}{}

\crefname{Problem}{Problem}{Problems}
\crefformat{Problem}{Problem~(#2#1#3)}
\crefrangeformat{Problem}{Problems~(#3#1#4) to~(#5#2#6)}

\crefname{Theorem}{Theorem}{Theorems}
\crefname{Corollary}{Corollary}{Corollaries}
\crefname{Proposition}{Proposition}{Propositions}
\crefname{Lemma}{Lemma}{Lemmas}
\crefname{Definition}{Definition}{Definitions}
\crefname{Example}{Example}{Examples}
\crefname{Assumption}{Assumption}{Assumptions}
\crefname{Remark}{Remark}{Remarks}
\crefname{Rem}{Remark}{Remarks}
\crefname{remarks}{Remarks}{Remarks}
\crefname{Appendix}{Appendix}{Appendices}
\crefname{Exercise}{Exercise}{Exercises}
\crefname{Theorem_A}{Theorem}{Theorems}
\crefname{Corollary_A}{Corollary}{Corollaries}
\crefname{Proposition_A}{Proposition}{Propositions}
\crefname{Lemma_A}{Lemma}{Lemmas}
\crefname{Definition_A}{Definition}{Definitions}

\usepackage{crossreftools}
\usepackage{algorithm,algorithmic}

\ifx\loadbreqn\undefined
	\relax
\else
	\usepackage{breqn} 
\fi

\ifx\renewtheorem\undefined
\ifx\useTheoremCounter\undefined
\newtheorem{Theorem}{Theorem}
\newtheorem{Corollary}{Corollary}
\newtheorem{Proposition}{Proposition}
\newtheorem{Lemma}{Lemma}
\else
\newtheorem{Theorem}{Theorem}

\newtheorem{Proposition}[theorem]{Proposition}
\fi

\newtheorem{Assumption}{Assumption}

\newtheorem{Lemma_A}{Lemma}[section]

\fi

\theoremstyle{remark}

\theoremstyle{plain}

\newenvironment{remarks}{
	\begin{list}{\textit{Remark} \arabic{Rem}:~}{
    \setcounter{enumi}{\value{Rem}}
    \usecounter{Rem}
    \setcounter{Rem}{\value{enumi}}
    \setlength\labelwidth{0in}
    \setlength\labelsep{0in}
    \setlength\leftmargin{0in}
    \setlength\listparindent{0in}
    \setlength\itemindent{15pt}
		}
}{
	\end{list}
}

\newcommand{\calA}{\mathcal{A}}
\newcommand{\calB}{\mathcal{B}}

\newcommand{\calD}{\mathcal{D}}

\newcommand{\calF}{\mathcal{F}}
\newcommand{\calG}{\mathcal{G}}
\newcommand{\calH}{\mathcal{H}}
\newcommand{\calI}{\mathcal{I}}

\newcommand{\calM}{\mathcal{M}}
\newcommand{\calN}{\mathcal{N}}

\newcommand{\calP}{\mathcal{P}}

\newcommand{\calR}{\mathcal{R}}

\newcommand{\calV}{\mathcal{V}}
\newcommand{\calW}{\mathcal{W}}
\newcommand{\calX}{\mathcal{X}}
\newcommand{\calY}{\mathcal{Y}}
\newcommand{\calZ}{\mathcal{Z}}

\DeclareSymbolFont{bsfletters}{OT1}{cmss}{bx}{n}
\DeclareSymbolFont{ssfletters}{OT1}{cmss}{m}{n}
\DeclareMathSymbol{\bsfGamma}{0}{bsfletters}{'000}
\DeclareMathSymbol{\ssfGamma}{0}{ssfletters}{'000}
\DeclareMathSymbol{\bsfDelta}{0}{bsfletters}{'001}
\DeclareMathSymbol{\ssfDelta}{0}{ssfletters}{'001}
\DeclareMathSymbol{\bsfTheta}{0}{bsfletters}{'002}
\DeclareMathSymbol{\ssfTheta}{0}{ssfletters}{'002}
\DeclareMathSymbol{\bsfLambda}{0}{bsfletters}{'003}
\DeclareMathSymbol{\ssfLambda}{0}{ssfletters}{'003}
\DeclareMathSymbol{\bsfXi}{0}{bsfletters}{'004}
\DeclareMathSymbol{\ssfXi}{0}{ssfletters}{'004}
\DeclareMathSymbol{\bsfPi}{0}{bsfletters}{'005}
\DeclareMathSymbol{\ssfPi}{0}{ssfletters}{'005}
\DeclareMathSymbol{\bsfSigma}{0}{bsfletters}{'006}
\DeclareMathSymbol{\ssfSigma}{0}{ssfletters}{'006}
\DeclareMathSymbol{\bsfUpsilon}{0}{bsfletters}{'007}
\DeclareMathSymbol{\ssfUpsilon}{0}{ssfletters}{'007}
\DeclareMathSymbol{\bsfPhi}{0}{bsfletters}{'010}
\DeclareMathSymbol{\ssfPhi}{0}{ssfletters}{'010}
\DeclareMathSymbol{\bsfPsi}{0}{bsfletters}{'011}
\DeclareMathSymbol{\ssfPsi}{0}{ssfletters}{'011}
\DeclareMathSymbol{\bsfOmega}{0}{bsfletters}{'012}
\DeclareMathSymbol{\ssfOmega}{0}{ssfletters}{'012}

\newcommand{\bbeta}{\bm{\beta}}
\newcommand{\bgamma}{\bm{\gamma}}

\newcommand{\bmu}{\bm{\mu}}

\newcommand{\bsigma}{\bm{\sigma}}

\newcommand{\bpsi}{\bm{\psi}}

\newcommand{\bPhi}{\bm{\Phi}}

\newcommand{\hsigma}{\widehat{\sigma}}

\DeclareMathOperator{\st}{s.t.}

\DeclareMathOperator{\diag}{diag}

\DeclareMathOperator{\Tr}{Tr}

\DeclareMathOperator{\rank}{rank}

\DeclareMathOperator{\vect}{vec}

\DeclarePairedDelimiter\parens{(}{)}
\DeclarePairedDelimiter\brk{[}{]}
\DeclarePairedDelimiter\braces{\{}{\}}

\newcommand{\qednew}{\nobreak \ifvmode \relax \else
      \ifdim\lastskip<1.5em \hskip-\lastskip
      \hskip1.5em plus0em minus0.5em \fi \nobreak
      \vrule height0.75em width0.5em depth0.25em\fi}

\newcommand{\nn}{\nonumber\\}

\newcommand{\T}{^{\intercal}}

\newcommand{\bone}{\mathbf{1}}

\newcommand{\norm}[1]{{\left\lVert{#1}\right\rVert}}
\newcommand{\trace}[1]{{\Tr\left( #1 \right)}}
\newcommand{\col}[1]{\operatorname{col}\left\{{#1}\right\}}
\newcommand{\row}[1]{\operatorname{row}\left\{{#1}\right\}}

\DeclareDocumentCommand\set{s m t| m}{%
  \IfBooleanTF#1%
	{\left\{\, #2\mathrel{} \IfBooleanTF{#3}{\middle|}{:}\mathrel{}  #4\, \right\}}%
  {\{\, #2 \IfBooleanTF{#3}{\mid}{\mathrel{} : \mathrel{}} #4\, \}}%
}

\DeclareDocumentCommand \ifcond {m m} {%
	{#1} %
	\IfValueT{#2}{\, \middle|\, {#2}}%
}

\DeclareDocumentCommand \P {e{_} g >{\SplitArgument{ 1 }{ @| }}d() g } {%
	\mathbb{P}%
	\IfValueTF{#1}{_{#1}}
		{\IfValueT{#2}{_{#2}}}%
	\IfValueT{#3}{\left(\ifcond#3}%
	\IfValueT{#4}{\, \middle|\, {#4}}%
	\IfValueT{#3}{\right)}%
}

\DeclareDocumentCommand \E {e{_} g >{\SplitArgument{ 1 }{ @| }}o g } {%
	\mathbb{E}%
	\IfValueTF{#1}{_{#1}}
		{\IfValueT{#2}{_{#2}}}%
	\IfValueT{#3}{\left[\ifcond#3}%
	\IfValueT{#4}{\, \middle|\, {#4}}%
	\IfValueT{#3}{\right]}%
}

\definecolor{gray90}{gray}{0.9}

\ifx\nohighlights\undefined

	\newcommand{\msout}[1]{\text{\color{green} \sout{\ensuremath{#1}}}}
	\newcommand{\del}[1]{{\color{green}\ifmmode \msout{#1}\else\sout{#1}\fi}}
\else

	\newcommand{\msout}[1]{#1}
	\newcommand{\del}[1]{#1}
\fi

\newcommand{\hhide}[1]{}

\graphicspath{{./Figures/}} 
\ifx\diagnoselabel\undefined
	\relax
\else
	\makeatletter
	 \def\@testdef #1#2#3{%
		 \def\reserved@a{#3}\expandafter \ifx \csname #1@#2\endcsname
		\reserved@a  \else
	 \typeout{^^Jlabel #2 changed:^^J%
	 \meaning\reserved@a^^J%
	 \expandafter\meaning\csname #1@#2\endcsname^^J}%
	 \@tempswatrue \fi}
	\makeatother
\fi
